\pdfoutput=1
\documentclass[a4paper,english,11pt]{article}
\usepackage[T1]{fontenc}
\usepackage[latin9]{inputenc}
\usepackage[bindingoffset=0.3cm,textheight=23cm,hdivide={2.3cm,*,2.3cm}, vdivide={*,23cm,*}]{geometry}
\usepackage{babel}
\usepackage{refstyle}
\usepackage{float}
\usepackage{calc}
\usepackage{mathrsfs}
\usepackage{amsmath}
\usepackage{amsthm}
\usepackage{amssymb}
\usepackage{dsfont}
\usepackage{upgreek}
\usepackage{slashed}
\usepackage{algorithm}
\usepackage{algpseudocode}
\usepackage{stmaryrd}
\usepackage{graphicx}
\usepackage{verbatim} 

\usepackage[unicode=true,
 bookmarks=true,bookmarksnumbered=false,bookmarksopen=false,
 breaklinks=true,pdfborder={0 0 1},backref=false,colorlinks=false]
 {hyperref}
\hypersetup{pdftitle={{Quasi-Coherent States and Quantum (Matrix) Geometry}},
 pdfauthor={Harold C. Steinacker}}

 \makeatletter

\newtheorem{thm}{Theorem}[section]
\newtheorem{lem}[thm]{Lemma}
\newtheorem{conj}{Conjecture}

\def\R{{\realv}} \def\C{{\compv}} \def\N{{\naturaln}}
\def\Z{{\integer}}

   \newcommand{\eq}[1]{(\ref{#1})}
\def\del{\partial}
\def\one{\mbox{1 \kern-.59em {\rm l}}}

\def\a{\alpha}

\def\d{\delta}
\def\l{\lambda}
\def\L{\Lambda}

\def\g{\gamma}

\def\nn{\nonumber}

\def\cM{{\cal M}}  
\def\cH{{\cal H}}  
  
\def\cB{{\cal B}}   
   
\def\cA{{\cal A}}  
\def\cC{{\cal C}}  
  
\def\cJ{{\cal J}}  
  
\def\cX{{\cal X}}  
\def\cN{{\cal N}}  
\def\cK{{\cal K}}

\def\cQ{{\cal Q}}

\def\mg{\mathfrak{g}}

\def\msu{\mathfrak{su}}
\def\mso{\mathfrak{so}}

\global\long\def\ket#1{\left|#1\right\rangle }

\global\long\def\realv{\mathds{R}}
\global\long\def\compv{\mathds{C}}
\global\long\def\integer{\mathds{Z}}
\global\long\def\naturaln{\mathds{N}}
\global\long\def\matunity{\mathds{1}}

\begin{document}

\renewcommand{\title}[1]{\vspace{11mm}\noindent{\Large{\bf #1}}\vspace{8mm}} 
\newcommand{\authors}[1]{\noindent{\large #1}\vspace{5mm}} 
\newcommand{\address}[1]{{\itshape #1\vspace{2mm}}}

%%%% --- TITLE PAGE --- %%%%
\begin{titlepage}
\begin{flushright}
 UWThPh-2020-22 %\\
\end{flushright}
\begin{center}
\title{ {\Large  Quantum (Matrix) Geometry  and  Quasi-Coherent States} }

\vskip 3mm

\authors{Harold C. Steinacker{\footnote{harold.steinacker@univie.ac.at}}
}
 
\vskip 3mm

 \address{ 

{\it Faculty of Physics, University of Vienna\\
Boltzmanngasse 5, A-1090 Vienna, Austria  }  
  }

\bigskip

\vskip 1.4cm

%%%% --- ABSTRACT --- %%%%
\textbf{Abstract}
\vskip 3mm

\begin{minipage}{14.5cm}%

A general framework is described which associates geometrical structures 
to any set of $D$ finite-dimensional hermitian matrices $X^a, \ a=1,...,D$.
This framework generalizes and systematizes the well-known examples of fuzzy spaces,
and allows to extract the underlying classical space without 
requiring the limit of large matrices or representation theory.
The approach is based on the previously introduced concept of 
quasi-coherent states.
In particular, a  concept of quantum K\"ahler geometry arises naturally, 
which includes the well-known quantized coadjoint orbits such as the fuzzy sphere
$S^2_N$ and fuzzy $\mathbb{C} P^n_N$.
A  quantization map for quantum K\"ahler geometries is established.
Some examples of quantum geometries which are not K\"ahler are identified,
including the minimal fuzzy torus.

\end{minipage}

\end{center}

\end{titlepage}

\tableofcontents{}

\section{Introduction}

It is expected on general grounds that the classical description 
of space-time geometry
is modified at very short length scales through quantum effects.
An interesting approach towards quantum geometry is based on 
quantized symplectic spaces, whose structure is similar to quantum mechanical 
phase space. Many examples of this type have been studied, 
starting with the fuzzy sphere $S^2_N$ 
\cite{Madore:1991bw,hoppe1982QuaTheMasRelSurTwoBouStaPro}, 
the fuzzy torus $T^2_N$ and more elaborate 2-dimensional 
spaces \cite{Arnlind:2010kw}, self-intersecting spaces such as 
squashed $\C P^2$ \cite{Steinacker:2014lma}, and many more.
A general class class is provided by quantized 
coadjoint orbits of compact semi-simple Lie groups. Many  classical features of the underlying symplectic space are encoded 
in their quantized version, which is based on the algebra of matrices $End(\cH)$
acting on a finite-dimensional Hilbert space $\cH$.

Of course, the notion of an algebra is not sufficient to define a geometry,
which should also contain a metric structure.
This extra structure arises in the context of Yang-Mills matrix models such as the IIB or IKKT model \cite{Ishibashi:1996xs}, which define a gauge theory on such fuzzy spaces.
In this context, a fuzzy space is specified by a  {\em set of hermitian matrices}
$X^a$ for $a=1,...,D$. These matrices not only generate the algebra of ``functions''
$End(\cH)$, but also naturally define a matrix Laplacian $\Box = \d^{ab} [X_a,[X_b,.]] $,
and  a Dirac-type operator $\slashed{D} = \Gamma_a[X^a,.]$ where $\Gamma_a$ are 
suitable Clifford or Gamma matrices. 
However rather than focusing on the spectral geometry\footnote{see e.g. \cite{Glaser:2019lcd} for related work in that context.} 
as in \cite{connes1995noncommutative}, 
we will emphasize a more direct approach based on 
(quasi-) coherent states defined through the matrices $X^a$, 
which provide a direct access to an underlying space $\cM$.

The obvious question is how to recover or extract the classical 
geometry underlying these quantized or ``fuzzy'' spaces, defined by the 
matrices $X^a$. 
For special cases such as quantized 
coadjoint orbits, one can construct {\em a sequence} of similar matrices
$X^a_{(N)} \in End(\cH_N)$, and show that the commutative description is recovered in the 
limit $N\to\infty$. 
This has led to the  attitude that the geometrical 
content of  fuzzy spaces can  only be obtained 
in some semi-classical  limit $N\to\infty$. 
%%indeed it might seem  unreasonable  to extract 
%geometry e.g. from some given set of three $2\times 2$ matrices $X^a$.
However, such a limit is not  satisfactory 
from a physics point of view, where one would like to 
attach geometrical meaning to a given set of matrices $X^a$.
In particular, this is required to interpret numerical simulations of 
Yang-Mills matrix models \cite{Nishimura:2019qal,Aoki:2019tby,Kim:2011cr,Anagnostopoulos:2020xai}, which are viewed as candidates for a quantum
theory of space-time and matter.

The purpose of the present paper is to establish a 
natural framework of ``quantum geometry'', which can be 
 associated to {\em any} given set of $D$ hermitian matrices
without requiring any limit, and which may admit a {\em semi-classical} 
or {\em almost-local} description 
in some {\em regime}.
This is based on the previously introduced concept of quasi-coherent states 
\cite{Ishiki:2015saa,Schneiderbauer:2016wub},
which can be associated to any set of hermitian matrices. The  concept is 
very much string-inspired \cite{Berenstein:2012ts}, 
and the quantum geometries are naturally viewed as 
varieties or ``branes'' embedded in target space. Since the mathematical concepts
are very close to those of quantum mechanics, the name
``quantum geometry'' seems justified, even if that name is perhaps already over-loaded with 
different meanings in the literature. 

The framework nicely captures the standard examples of fuzzy spaces, but it is completely general. 
Moreover, it naturally leads to an intrinsic concept of quantum K\"ahler geometry, 
which is a special class of quantum geometries which satisfy certain
conditions\footnote{This is in distinct from 
the approach in  \cite{Ishiki:2016yjp}.}; there is no need to add any  
structure by hand. 
Of course for quantized coadjoint orbits, the coherent states are obtained easily 
from the representation theory.
However, the present construction allows to reconstruct the full K\"ahler structure 
of the (quantum) space {\em without resorting to representation theory},
which is not known in general.

In the semi-classical limit, many of the structures and steps 
have been considered before,
notably in work by Ishiki etal \cite{Ishiki:2015saa,Ishiki:2016yjp} and in \cite{Schneiderbauer:2016wub,Berenstein:2012ts,deBadyn:2015sca,Karczmarek:2015gda}. 
However, the novelty is in introducing a more abstract point of view.
We introduce the concept of an {\em abstract quantum space} $\cM$, 
by considering the space of quasi-coherent states as a sub-variety
of $\C P^N$. This allows to make concise statements
for finite $N$, and to give a clear conceptual correspondence 
between finite matrix configurations and geometry,
based on a space $Loc(\cH)\subset End(\cH)$ of almost-local operators. The semi-classical 
description applies in some infrared (IR) regime, while
the UV regime of matrix geometry displays a very different and stringy nature, 
which is manifest in string states.
This  framework also allows to establish the existence of a 
surjective quantization map for quantum K\"ahler manifolds, and to 
make some non-trivial regularity statements about the
abstract quantum space $\cM$.

It is important to note that 
the proposed framework is more than just some ad-hoc procedure:
by definition, the quasi-coherent states 
provide an optimal basis where all matrices have minimal joint uncertainty,
i.e. they are simultaneously ''almost-diagonal``.
Such almost-commuting configurations are expected to play a dominant role in Yang-Mills matrix models. The approach is well-suited to be implemented on a 
computer \cite{lukas_schneiderbauer_2016_45045,lukas_schneiderbauer_2019_2616687}, and 
should provide a powerful tool to 
understand and interpret the results of numerical simulations
of Yang-Mills matrix models.
%\footnote{This was successfully applied to 
%some early numerical simulations \cite{gutleb} kindly provided 
%to the author by  J. Nishimura and A. Tsuchiya.},
% improving upon some more ad-hoc approaches based on block-matrices.

This paper comprises 3 main parts. In section \ref{sec:quasicoherent}
we define the quasi-coherent states $|x\rangle$ for  $x\in\R^D$,
and study their properties as functions of $x\in\R^D$.
Much of this section is more-or-less known in some form, but at least the 
relation with solutions of the matrix-Yang-Mills equation is new.
In section \ref{sec:abstract-quantum}, we  introduce the central concept of 
an abstract quantum space $\cM\subset\C P^{N}$. This offers 
a conceptually clear definition of almost-local operators
and the semi-classical regime. It also leads to a natural 
concept of a real and complex quantum tangent space and quantum K\"ahler 
manifolds. Some consequences are developed in section \ref{sec:coherent},
notably a quantization map for quantum K\"ahler manifolds.
These concepts are  illustrated in a number of 
examples in section \ref{sec:examples},  and
the relation with physical Yang-Mills matrix models is briefly discussed in 
section \ref{sec:remarks}.

%% more cit:  oConnor ?

\section{Quasi-coherent states on $\R^D$}
\label{sec:quasicoherent}

In this paper, a {\bf matrix configuration} will be a collection of $D$ hermitian matrices $X^a \in End(\cH)$
acting on some (separable) Hilbert space $\cH$. To avoid technical complications, we will assume that 
$\cH\cong\C^N$ is finite-dimensional, apart from some illustrative infinite-dimensional examples.
Such a matrix configuration will be called {\bf irreducible} if  the only matrix which commutes
with all $X^a$ is the unit matrix. Equivalently, the algebra generated by the 
$X^a$ is the full matrix algebra $End(\cH)$. This will be assumed throughout.

%%HS added
By definition, such an irreducible matrix configuration does not admit
any common eigenvectors $|\psi\rangle$, since otherwise $|\psi\rangle\langle \psi|$
would commute with all $X^a$. Nevertheless, we are mainly interested in matrix configurations which are ''almost-commuting``, in the sense that the 
commutators $[X^a,X^b]$ are ''small'';  these are expected
to be the dominant configurations in Yang--Mills matrix models such as the 
IIB or IKKT model \cite{Ishibashi:1996xs}. We therefore wish to find 
a set of states which are optimally adapted to the matrix configuration,
so that the $X^a$ are ''as diagonal as possible``. This may also be of interest 
in different contexts.

With this in mind, we 
associate to an irreducible matrix configuration $X^a$ and a point $x\in\R^D$
the following 
{\bf displacement Hamiltonian}\footnote{
%%HS V2
As explained in section \ref{sec:eff-metric-MM}, $H_x$
can be interpreted in the IIB model as energy of a point--brane at $x$ 
on the background defined by the matrix configuration $X^a$.} 
(cf. \cite{Ishiki:2015saa,Schneiderbauer:2016wub})
\begin{align}
 H_x := \frac 12\sum_{a=1}^D\, (X^a-x^a\one)^2 \ .
\end{align}
This is a  positive definite\footnote{To see positive-definiteness, assume that 
$H_x|\psi\rangle = 0$; this implies $X^a |\psi\rangle =x^a|\psi\rangle$ for all $a$,
but then $[H_x,|\psi\rangle\langle\psi|] = 0$ in 
contradiction with irreducibility.} hermitian operator  on $\cH$,
which can be thought of as an analog to the shifted harmonic oscillator.
%%HS V2 added
It allows to find optimally localized approximate eigenstates 
for the given matrix configuration as follows.
Let $\l(x) > 0$ be the lowest eigenvalue of $H_x$. 
A {\bf quasi-coherent state}  $|x\rangle$ at $x$ is then defined following \cite{Ishiki:2015saa,Schneiderbauer:2016wub} as normalized vector
$\langle x|x\rangle = 1$
in the eigenspace $E_x$ of $H_x$ with eigenvalue $\l(x)$, 
\begin{align}
 H_x |x\rangle = \lambda(x) |x\rangle \ .
 \label{quasicoh-def}
\end{align}
We will assume for simplicity that $E_x$ is one-dimensional, 
except possibly on some singular set $\cK\subset\R^D$.
Clearly the quasi-coherent states $|x\rangle$ form a $U(1)$ bundle 
\begin{align}
 \cB  \ \to \tilde\R^D \qquad \mbox{over} \quad \tilde\R^D := \R^D \setminus\cK \ .
 \label{sing-set-def}
\end{align}
% which we can alternatively view as line bundle
% \begin{align}
%  \tilde\cB \to \tilde\R^D \qquad \mbox{over} \quad \tilde\R^D \ .
%  \label{sing-set-def}
% \end{align}
Standard theorems \cite{rellich1969perturbation,kato2013perturbation}
ensure that $\l(x)$ and $E_x$ depend smoothly
on $x\in\tilde\R^D$. 
We can then choose some local section  of $\cB$ near any given point $\xi\in\tilde\R^D$, 
denoted by $|x\rangle$.
Thus $\cK$ is the location where different 
eigenvalues of $H_x$ become degenerate. 
If $\l(x)$ can be extended smoothly at some point $p\in\cK$,
different eigenvalues simply touch without crossing, and
the sections $|x\rangle$ and the bundle
$\cB$ can be extended through $p$; 
we can then basically remove $p$ from $\cK$.
Hence we can assume that $\cK$ contains  
only points where some eigenvalues cross, i.e. $\l(x)$ cannot be 
continued.
We denote this $\cK$ as  {\bf singular set}.
The bundle is non-trivial if $\cK\neq 0$.

For any operator in $\Phi\in End(\cH)$, we can define the {\bf symbol} in $\cC(\tilde\R^D)$ through the map 
\begin{align}
 End(\cH) & \to \cC(\tilde\R^D)  \nn\\
 \Phi &\mapsto \langle x|\Phi|x\rangle =: {\bf\phi}(x) \ .
 \label{symbol-RD}
\end{align}
Elements of $End(\cH)$ will be indicated by 
upper-case letters, and functions by lower-case letters.
The map \eq{symbol-RD}
should be viewed as a de-quantization map, associating classical functions to 
noncommutative ``functions'' (or rather observables) in $End(\cH)$.
In particular, the symbol of the matrices $X^a$ provides a map
\begin{align}
   {\bf x^a} :\quad \tilde\R^D &\to \R^D  \\
    x &\mapsto {\bf x^a}(x):= \langle x|X^a|x\rangle \ .
    \label{X-ecpect-embedding-symbol}
\end{align}
%(cf. (7.1) in \cite{Schneiderbauer:2016wub}).
Generically ${\bf x^a}(x) \neq x^a$, and the deviation is 
measured by the {\bf displacement}
\begin{align}
 d^2(x) := \sum_a({\bf x^a}(x) - x^a)^2 \ .
\end{align}
The quality of the matrix configuration (or of the underlying quantum space)
is measured by 
the {\bf dispersion} or uncertainty
\begin{align}
\d^2(x) &:= \sum_a (\Delta X^a)^2   \nn\\
 (\Delta X^a)^2 &:= \langle x|(X^a - {\bf x^a}(x))^2|x\rangle 
  = \langle x|X^a X^a|x\rangle
  - {\bf x^a}(x) {\bf x^a}(x) \ .
  \label{dispersion}
\end{align}
If $\d^2(x)$ is small, then the $X^a$ can be interpreted as operators or observables which approximate 
the functions ${\bf x^a}$ on $\tilde\R^D$, and if $d^2(x)$ is also small then $X^a \approx {\bf x^a} \approx x^a$.
Note that \eq{quasicoh-def} implies 
\begin{align}
 \l(x) = \d^2(x) + d^2(x) \ ,
 \label{l-delta-D}
\end{align}
hence a small $\l(x)$ implies that both $\d^2(x)$ and $d^2(x)$ are bounded  by $\l(x) > 0$.
% We will denote a regime where $\d^2\approx 0$ is negligible in some given context 
% as ``semi-classical regime''; this will be discused further below. 
$d^2(x)$ will be understood  in section \ref{sec:abstract-quantum} as  displacement 
of $x$ from the underlying quantum space or brane $\cM$.
Hence quasi-coherent states should be viewed as the states 
with minimal dispersion and displacement for given $x\in\tilde\R^D$,
cf. \cite{Schneiderbauer:2016wub} for a more detailed discussion.

\subsection{$U(1)$ connection, would-be symplectic form 
and quantum metric} 
\label{sec:inner-connect-symp}

Now we associate to any matrix configuration two unique tensors on $\tilde\R^D$: the {\em would-be symplectic form} $\omega_{ab}$ 
and {\em quantum metric} $g_{ab}$.
Since $|x\rangle \in\cH$, the bundle  $\cB$ over $\tilde\R^D$ naturally inherits a metric and a connection. 
We can define a connection 1-form $A$  via
\begin{align}
 P \circ d |x\rangle  &= |x\rangle i A , \qquad 
 iA := \langle x|d|x\rangle \quad \in\Omega^1(\tilde\R^D)
 \label{nabla-def}
\end{align}
where $P= |x\rangle\langle x|$ is the projector on $E_x$. Here $A$ is real because 
\begin{align}
 (\langle x |d|x\rangle)^* = d(\langle x |) |x\rangle = - \langle x |d|x\rangle \ ,
\end{align}
and transforms like a $U(1)$ gauge field 
\begin{align}
 |x\rangle \to e^{i\L(x)}|x\rangle, 
 \qquad A_a \to A_a +  \del_a \L \ .
\end{align}
In particular, we can 
parallel transport $|x\rangle$  along a path
$\gamma$ in $\tilde\R^D$. 
This connection is analogous to a Berry connection. It
is encoded in the inner product
\begin{align}
 \langle x|y\rangle =: e^{i \varphi(x,y) - D(x,y)} \ ,
 \label{coherent-inner}
\end{align}
which defines a distance function $D(x,y)$ and a phase function $\varphi(x,y)$
which satisfy
\begin{align}
 D(x,y) &= D(y,x) \geq 0, \qquad D(x,y) = 0 \ \Leftrightarrow \  x=y \nn\\
 \varphi(x,y) &= - \varphi(y,x) \ .
\end{align}
The phase clearly depends on the particular section $|x\rangle$ of the bundle $\cB$, while $D(x,y)$ does not.
To understand these two functions, we 
differentiate \eq{coherent-inner} w.r.t. $y$ 
\begin{align}
 \langle x |d_y|y\rangle|_{y=x} =  i d_y \varphi(x,y)|_{y=x} - d_y D(x,y)|_{y=x} \ .
\end{align}
Comparing with \eq{nabla-def} we conclude 
\begin{align}
i d_y \varphi(x,y)|_{y=x} &= i A = i A_a dx^a \nn\\
 d_y D(x,y)|_{y=x} &= 0   \ .
 \label{del-D-varphi}
\end{align}
Hence the phase $\varphi(x,y)$ encodes the connection
$A$.
For a contractible closed path $\g = \del\Omega$ in $\tilde\R^D$,
the change of the phase of $|x\rangle$ along $\g$
is hence given by the field strength via Stokes theorem
\begin{align}
 \oint_\g A = \int_\Omega dA\ .
 \label{flux-integral}
\end{align}
If the connection is flat, the phase $\varphi(x,y)$ can be gauged away completely.

To proceed, consider the gauge-invariant hermitian $D\times D$ matrix defined by
\begin{align}
 h_{ab} &= \big((\del_a +i A_a) \langle x|\big)
 (\del_b - i A_b)|x\rangle|_{\xi} = h_{ba}^*  \nn\\
 % &= (\del_{x^a}+i A_a)(\del_{y^b}- i A_b) \langle x|y\rangle|_{\xi} \nn\\ 
&= (\del_{x^a}+i A_a)(\del_{y^b}- i A_b) e^{i \varphi(x,y) - D(x,y)}|_{\xi} \nn\\
   &=: \frac i2(\omega_{ab} + g_{ab})
   \label{hab-def}
\end{align}
 at some reference point $\xi\in\tilde\R^D$, which decomposes into 
 the real symmetric and antisymmetric tensors 
 $g_{ab}$ and $\omega_{ab}$.
The symmetric part  
\begin{align}
  g_{ab} &= \big((\del_a +i A_a) \langle x|\big)
 (\del_b - i A_b)|x\rangle + (a\leftrightarrow b) \nn\\
  &= (\del_a\langle x|)\del_b |x\rangle - A_a A_b+ (a\leftrightarrow b)
%  false!  &= (\del_a\langle x|)\del_b |x\rangle + (a\leftrightarrow b)
\label{g-def}
\end{align}
(using\eq{nabla-def})
is the pull-back of the Riemannian metric\footnote{Note that $g_{ab}$ is not related to the Euclidean metric $\d_{ab}$ on 
target space  $\R^D$. } on $\cH$ 
(or equivalently of the Fubini--Study metric on $\C P^{N-1}$) through 
the section $|x\rangle$.
The antisymmetric part of $h_{ab}$ encodes a 2-form 
\begin{align}
 i\omega_{ab} &= i(\del_a A_b - \del_b A_a)
   = (\del_{a} \langle x|)\del_{b}|x\rangle 
   - (\del_{b} \langle x|)\del_{a}|x\rangle \nn\\
 i\omega = \frac i2\omega_{ab} dx^a \wedge dx^b \
  &= d\langle x|\wedge d|x\rangle
  =  d(\langle x |d|x\rangle) = i dA
 \label{omega-def}
\end{align}
which is the $U(1)$ field strength of the connection $A$ and therefore closed,
\begin{align}
 \omega = dA , \qquad   d\omega = 0 \ .
 \label{F-omega-dA}
\end{align}
Assuming (local) translation invariance\footnote{Translational invariance holds at sufficiently short scales, cf. \eq{parallel-transport-states-loc}.},
it follows that the expansion of $\varphi(x,y)$ to quadratic order in $x$ and $y$ (setting $\xi = 0$)  is
\begin{align}
 \varphi(x,y) &=  A_a (y^a - x^a) - \frac 14 \omega_{ab}(x-y)^a (x-y)^b + ... \ .
 \end{align}
Similarly,  the expansion of $D(x,y)$ 
is given by 
\begin{align} 
 D(x,y) &= \frac 14 (x-y)^a (x-y)^b  g_{ab} + ...
  \label{D-expand}
\end{align}
since $D(x,y)$ is gauge invariant and satisfies $D(x,x) = 0$ and $D(x,y)\geq 0$.
In fact viewing $\cB/U(1)$ as subset of 
$\C P^{N-1}$, we can use the well-known formula
\begin{align}
 \cos^2(\g(x,y)) = e^{-2D(x,y)}
\end{align}
where $\g(x,y)$
is the geodesic distance squared between $|x\rangle$ and $|y\rangle$
 in the Fubini--Study metric on  $\C P^{N-1}$.
Combining \eq{coherent-inner} and \eq{D-expand}, 
we learn that the quasi-coherent states
are localized within a region of size 
\begin{align}
L_{\rm coh}^2 = \|g_{ab}\|^{-1}
% g_{ab}\Delta x^a\Delta x^b = 1 \ , \qquad \Delta x^a
 \label{coherence-scale}
\end{align}
denoted as {\bf coherence scale}.
The $|x\rangle$ are approximately constant
below this scale due to \eq{coherent-inner}. The relation with 
the uncertainty of $X^a$ will be given in \eq{uncertainty-X}.
% Since the quasi-coherent states are by definition optimally localized,
% this characterizes the scale of 
% uncertainty or the ``fuzzyness'' of the geometry;
Therefore $g_{ab}$ will be denoted as {\bf quantum metric}.
We will see in section \ref{sec:eff-metric-MM}
that there is a different, {\em effective} metric
which governs the low-energy physics on such quantum spaces in 
Yang-Mills matrix models. 
However, the intrinsic structure of the underlying quantum space
is best understood using a more abstract  point of view 
developed in section \ref{sec:abstract-quantum}.

We will see that $\omega$ typically
arises from a symplectic form on an underlying space $\cM$. 
Therefore $\omega$ will be denoted as {\bf would-be symplectic form}.
Since it is the curvature of a $U(1)$ bundle,
its flux is quantized
for every  2-cycle $S^2$ in $\tilde\R^D$ as
\begin{align}
 \int\limits_{S^2} \frac1{2\pi}\omega = n, \qquad n\in\Z \ .
 \label{quant-cond-S2}
\end{align}
This 
arises using \eq{flux-integral} as consistency condition on the $U(1)$ holonomy for the parallel transport 
along a closed path $\g$ on $S^2$.
%Analogous quantization conditions apply for all higher classes 
%$\omega^{\wedge k} = d(A\wedge\omega^{\wedge k-1})$.
In more abstract language,  $c_1 = -\frac{1}{2\pi} \omega$ is
the first Chern class of  $\cB$ viewed as  line bundle,
which is the pull-back of the first Chern class (or symplectic form) 
 of $\C P^{N-1}$ 
 via the section $|x\rangle$.
 The  bundle $\cB$ is trivial if these
numbers vanish for all cycles $S^2$,
 hence if $H^2(\tilde\R^D)$ vanishes.

\subsection{Differential structure of quasi-coherent states}
\label{sec:derivatives-generators} 

Assume that $|x\rangle$ is a local section of the quasi-coherent states, with
\begin{align}
 H_x|x\rangle = \l(x) |x\rangle \ .
 \label{H-EV-2}
\end{align}
Using Cartesian coordinates $x^a$ on $\R^D$,
we observe that 
\begin{align}
 \del_a H_x = -(X_a-x_a \one) \ .
 \label{del-H-X}
\end{align}
Thus differentiating \eq{H-EV-2} gives 
\begin{align}
 (H_x - \l(x))\del_a|x\rangle &=  -\del_a (H_x - \l(x))|x\rangle
  =  \big(X_a-x_a + \del_a\l\big)|x\rangle   \ .
\label{deform-eigenstate}
\end{align}
Since lhs is orthogonal to $\langle x|$, it follows that
\begin{align}
 0 = \langle x|\big(X_a-x_a + \del_a\l\big)|x\rangle   
 \label{exp-X-1}
\end{align}
so that the expectation value or symbol of the basic matrices $X_a$ is given by
\begin{align}
\boxed{
 {\bf x_a} = \ \langle x| X_a|x\rangle  = x_a - \del_a\l \ .
 }
 \label{expect-X}
\end{align}
Furthermore, \eq{deform-eigenstate} gives
(in the non-degenerate case under consideration)
\begin{align}
  \del_a |x\rangle &= |x\rangle \langle x| \del_a |x\rangle
    + (H_x - \l)^{-1}\big(X_a-x_a + \del_a\l\big)|x\rangle \nn\\
 (\del_a - i A_a) |x\rangle &=  (H_x - \l)^{-1}\big(X_a-x_a + \del_a\l\big)|x\rangle 
 \label{del-nabla-X-general}
\end{align}
using \eq{del-D-varphi}.
Even though the $(H_x - \l)^{-1}$ term is well-defined  here,
it is better to replace $(H_x - \l)^{-1}$ with 
an operator that is well-defined on $\cH$.  This is achieved using
the ``reduced resolvent''
\begin{align}
 (H_x - \l(x))^{'-1} := (\one-P_x) (H_x - \l(x))^{-1}(\one-P_x), 
 \qquad \qquad P_x := |x\rangle\langle x|
\end{align}
which satisfies 
\begin{align}
 (H_x - \l) (H_x - \l)^{'-1} &=\one- P_x = (H_x - \l)^{'-1}(H_x - \l), \nn\\
 (H_x - \l)^{'-1} |x\rangle &= 0 \ .
 \label{gen-inv-rel}
\end{align}
Observing  $(H_x - \l)^{'-1} (x_a - \del_a\l)|x\rangle = 0$ due to \eq{H-EV-2}, we can
write \eq{del-nabla-X-general} as
\begin{align}
\boxed{ \
  (\del_a - i A_a) |x\rangle = i\cX_a |x\rangle
  \ }
  \label{del-nabla-X-general-2}
\end{align}
for $\cX_a = -i(H_x - \l)^{'-1} X_a$. Since $(H_x - \l)^{'-1}|x\rangle = 0$,
this can be replaced by the hermitian generator
\begin{align}
\boxed{ \
  \cX_a := -i[(H_x - \l)^{'-1}, X_a] =  \cX_a^\dagger  \ .
   \ }
  \label{cX-def}
\end{align}
Moreover, we note
\begin{align}
  \langle x|\cX_a|x\rangle = 0 \ .
  \label{cX-exp-zero}
\end{align}
Hence $\cX_a$ generates the gauge-invariant
tangential variations of $|x\rangle$, which take value in the orthogonal
complement of $|x\rangle$.
This will be the basis for defining the quantum tangent space in section \ref{sec:abstract-quantum}.
The local section $|x\rangle$ over $\tilde\R^D$ can now be written as 
\begin{align}
 |x\rangle = P \exp\Big(i\int_\xi^x (\cX_a + A_a) dx^a\Big)|\xi\rangle
 \label{parallel-transport-states}
\end{align}
near the reference point $\xi\in \tilde\R^D$.
Here $P$ indicates path ordering, which is just a formal way of writing the 
solution of \eq{del-nabla-X-general-2}.
 In a small local neighborhood,
the $\cX_a$ are approximately constant, and $A_a$ can be gauged away. Then
\eq{parallel-transport-states} can be written as
\begin{align}
 |x\rangle \approx e^{i(x-y)^a \cX_a}|y\rangle \ ,
 \label{parallel-transport-states-loc}
\end{align}
which means that the $\cX_a$ generate the local translations on $\cM$.

\subsection{Relating the algebraic and geometric structures}

Since the derivatives of $|x\rangle$ are 
spanned by the $\cX^a|x\rangle$, the $U(1)$
field strength $\omega_{ab}$ and the 
quantum metric $g_{ab}$ should be related to 
 algebraic properties for the $\cX^a$.
Indeed,
starting from \eq{hab-def}
\begin{align}
 h_{ab} &= \langle x|\cX_a \cX_b|x\rangle
 = \frac i2(\omega_{ab} + g_{ab}) \ ,
\end{align}
we obtain 
\begin{align}
 i\omega_{ab} = h_{ab} - h_{ba} 
  &= \langle x|(\cX_a \cX_b - \cX_b \cX_a) |x\rangle
   \label{om-ab-XX}
\end{align}
and 
\begin{align}
 g_{ab} &= h_{ab} + h_{ba} =\langle x|(\cX_a \cX_b + \cX_b \cX_a) |x\rangle \ .
 \label{gab-XX}
\end{align}
 This provides a first link between the geometric and  algebraic 
 structures under consideration.
Furthermore, is useful to define the following hermitian
tensor (similar as in \cite{Ishiki:2016yjp})
\begin{align}
 P_{ab}(x) &:= \langle x| X_a (H_x - \l)^{'-1}
  X_b|x\rangle = P_{ba}(x) ^* \nn\\
  &=  i\langle x| X_a \cX_b|x\rangle \  = -i\langle x| \cX_a X_b|x\rangle \ .
  \label{P-explicit-herm}
\end{align}
Its symmetric part is obtained by taking  derivatives of
\eq{expect-X} 
\begin{align}
  \del_b {\bf x_a}(x) &= \del_b x_a - \del_b\del_a\l \nn\\
  &= \del_b\langle x|X_a|x\rangle
   = i\langle x|[X_a,\cX_b]|x\rangle \nn\\
  % &= P_{a b} + P_{ab}^* +  i (A_b {\bf x_a} - A_b {\bf x_a})\nn\\
 &= P_{a b} + P_{ba} 
\label{projector-embed}
\end{align}
lowering indices with $\d_{ab}$; for the antisymmetric part see  \eq{imaginary-P}.
This will be recognized as projector on the 
embedded quantum space  in \eq{P-M-proj}, as obtained
in the semi-classical limit in \cite{Ishiki:2016yjp}.

\subsection{Almost-local operators}
\label{sec:almost-local}

We would like to define a class 
$Loc(\cH) \subset End(\cH)$ of {\bf almost-local operators}
which satisfy
\begin{align}
\Phi|x\rangle  &\approx 
 |x\rangle \langle x|\Phi|x\rangle = P_x \Phi|x\rangle \ 
  = |x\rangle \phi(x)\qquad \forall x\in\tilde\R^D 
 \label{Phi-loc-approx}
 \end{align}
 where $\phi(x) = \langle x|\Phi|x\rangle$ is the symbol of $\Phi$, and $P_x = |x\rangle\langle x|$ is the 
 projector on the quasi-coherent state $|x\rangle$.
 The question is how to make  the meaning of 
$\, \approx\, $ precise, without considering some limit  
 as in \cite{Ishiki:2015saa}. We should certainly require that 
$\Phi|x\rangle \approx |x\rangle \phi(x)$  in $\cH$ for every $x$,
but it is not obvious yet how to handle the dependence on $x$,
and how to specify bounds.
The guiding idea is that it should make sense to identify  
$\Phi$ with its symbol 
\begin{align}
 \Phi \ \sim  \  \phi(x) = \langle x|\Phi|x\rangle \ ,
\end{align}
 indicated by $\sim$ from now on.
This will be made more precise 
in the section \ref{sec:quatiz-semi} by requiring that $\sim$ is an {\em approximate isometry} from
$Loc(\cH)$ to $\cC_{\rm IR}(\cM)$, where
$\cC_{\rm IR}(\cM)$ is a class of ``infrared'' functions on the abstract 
quantum space associated to the matrix configuration.
%%HS V2 added
The essence of almost-locality is then that
the {\em integrated} deviations 
from classically are small compared with the classical 
values.
With this in mind, we  proceed to elaborate some  consequences
of \eq{Phi-loc-approx} for fixed $x$ without specifying bounds.

Since $(\one - P_x)$ is a projector, we have the estimate
\begin{align}
 \langle x|\Phi^\dagger\Phi|x\rangle 
  = \langle x|\Phi^\dagger P_x\Phi|x\rangle 
  + \langle x|\Phi^\dagger(\one - P_x)\Phi|x\rangle 
  &\geq \ \langle x|\Phi^\dagger P_x\Phi|x\rangle = |\phi(x)|^2 \ .
  \label{norm-loc-estim}
\end{align}
It follows that every hermitian almost-local operator $\Phi=\Phi^\dagger$ satisfies
\begin{align}
 \langle x|\Phi\Phi|x\rangle \approx \langle x|\Phi|x\rangle ^2
  = |\phi(x)|^2 \qquad \forall x\in\tilde\R^D\ ,
\label{uncertain-1}
\end{align}
i.e. the uncertainty
of $\Phi$ is negligible,
\begin{align}
 \langle x|(\Phi- \langle x|\Phi|x\rangle)^2|x\rangle \approx 0
 \qquad \forall x\in\tilde\R^D\ .
 \label{uncertain-Phi}
\end{align}
This means that  $(\Phi-\phi(x))|x\rangle$ is 
approximately zero, which in turn implies \eq{Phi-loc-approx}.
Therefore almost-locality is essentially equivalent to \eq{uncertain-Phi},
up to global considerations and specific bounds. 
%%HS V2
A more succinct global version of \eq{uncertain-Phi} is given in 
\eq{norm-uncertainty}.

We also note that for two operators $\Phi,\Psi\in Loc(\cH)$
the factorization properties 
\begin{align}
 \Phi\Psi|x\rangle
 &\approx \Phi|x\rangle\langle x|\Psi|x\rangle 
 \approx |x\rangle \phi(x)\psi(x)\nn\\
 \langle x|\Phi\Psi|x\rangle 
 &\approx \phi(x)\psi(x) 
 \label{semiclass-fact}
\end{align} 
follow formally. However this does not mean that $Loc(\cH)$ is an algebra,
since the specific bounds may be violated by the product.
For some given matrix configuration, $Loc(\cH)$ may be empty or very small. 
This happens e.g. for the minimal fuzzy spaces  as discussed in section \ref{sec:degenerate}, and it is 
expected for random matrix configuration. 
But even in these cases, the associated 
geometrical structures still provide useful insights.

For interesting quantum geometries, we expect that 
all the $X^a$ are almost-local,  hence 
also polynomials $P_n(X)$ up to some maximal degree 
$n$ due to \eq{semiclass-fact}. $Loc(\cH)$ 
can often be characterized by some bound on the 
eigenvalue of $\Box$ \eq{Box}, or the uncertainty scale 
$L_{NC}$ \eq{L-NC-def}. However, 
$Loc(\cH)$ can never be more than a small subset of $End(\cH)$.

\subsection{Almost-local quantum spaces and Poisson tensor}
\label{sec:almost-loc-space}

To see how the Poisson structure 
arises,  define the real anti-symmetric matrix-valued function
\begin{align}
 \theta^{ab} := -i\langle x| [X^a,X^b] |x\rangle = -\theta^{ba}
 \label{theta-def}
\end{align}
on $\tilde\R^D$.
To relate it to the previous structures, we shall loosely follow \cite{Ishiki:2016yjp}, starting from
the identity
\begin{align}
  [X^a,X^b](X_b-x_b) + (X_b-x_b)[X^a,X^b]
   % = [X^a,(X^b-x^b)(X_b-x_b)] 
   = 2 [X^a,H_x] \ . 
   \label{H-x-comm}
\end{align} 
Taking the expectation value, we obtain
\begin{align}
 \langle x| [X^a,X^b](X_b-x_b)|x\rangle + \langle x| (X_b-x_b)[X^a,X^b]  |x\rangle
  =  2 \langle x|[X^a,H_x] |x\rangle = 0 \ .
  \label{comm-X-rel}
\end{align}
If $X^a$ is almost-local\footnote{This is expected from the 
 definition of quasi-coherent states, as long as the uncertainty is 
sufficiently small.}, then this implies
\begin{align}
 0 \approx \langle x| [X^a,X^b] |x\rangle \langle x| (X_b-x_b)|x\rangle
  = - i\theta^{ab} \del_b\l
  \label{theta-del-l}
\end{align}
using \eq{expect-X}. In section \ref{sec:quatiz-semi}  we will see  that
this implies  $\l \sim {\rm const}$ on the embedded quantum space $\tilde\cM$,
and $P_{ac} + P_{ca} \sim \del_c x_a$ is its tangential projector.

We now define an {\bf almost-local quantum space} to be a matrix configuration where all $X^a$ as well as all $[X^a,X^b]$
are almost-local operators.  
Then they approximately commute, and we can proceed  following \cite{Ishiki:2016yjp}
\begin{align}
-2 (H_x - \l) (X^a -x^a + \del^a\l) |x\rangle
 &= 2 [X^a,H_x] |x\rangle
  \approx 2 (X_b-x_b)[X^a,X^b]  |x\rangle \nn\\
 &\approx 2(X_b-x_b) |x\rangle\langle x|[X^a,X^b]|x\rangle  
= 2i (X_b-x_b) |x\rangle\theta^{ab}\nn\\
&\approx 2i (X_b-x_b + \del_b\l) |x\rangle\theta^{ab}
\label{XX-asympt}
\end{align}
using  the factorization property, \eq{theta-del-l} and \eq{H-EV-2}.
However the first approximation is subtle, since $(X_b-x_b) |x\rangle \approx 0$.
This can be justified  if $X^a$ is a solution of the 
{\bf Yang-Mills equations}\footnote{This argument also goes through for 
the generalized Yang-Mills equation $\Box X^a \equiv [X_b,[X^b,X^a]] = m\, X^a$
as long as $m$ is sufficiently small, where $\Box$ is defined in \eq{Box}.}
\begin{align}
 [X_b,[X^b,X^a]] = 0 \ 
\end{align}
which are indeed the equations of motion for Yang-Mills matrix models \cite{Ishibashi:1996xs}.
Then \eq{H-x-comm}  implies
\begin{align}
  [X^a,H_x]
 % (X_b-x_b)[X^a,X^b] + \frac 12 [[X^a,X^b],X_b] \nn\\
  &= (X_b-x_b)[X^a,X^b]
\end{align}
and the above steps become
\begin{align}
-2 (H_x - \l) X^a |x\rangle
 &= 2 [X^a,H_x] |x\rangle
  = 2 (X_b-x_b)[X^a,X^b] |x\rangle  \nn\\
&\approx 2i (X_b-x_b) |x\rangle\theta^{ab} \nn\\
&\approx 2i (X_b-x_b + \del_b\l) |x\rangle\theta^{ab} \ .
\label{XX-asympt-2}
\end{align}
The rhs is indeed orthogonal to $\langle x|$ due to  \eq{expect-X}, and  we can conclude
\begin{align}
 - (X^a -x^a + \del^a\l) |x\rangle
 &\approx i(H_x - \l)^{'-1} (X_b-x_b+ \del_b\l) |x\rangle\theta^{ab} \nn\\
 &= -\theta^{ab}\cX_b |x\rangle
 = i \theta^{ab} (\del_b -iA_b)|x\rangle
\end{align}
hence
\begin{align}
 (X^a -x^a + \del^a\l)|x\rangle \approx - i \theta^{ab} (\del_b -iA_b)|x\rangle
  \label{Xtheta-1}
\end{align}
and by conjugating
\begin{align}
 \langle x| (X^d -x^d + \del^d\l)  \approx
 i \theta^{dc} (\del_c + iA_c)\langle x|  \ .
  \label{Xtheta-2}
\end{align}
These relations are very useful. 
First, they imply the important relation
\begin{align}
\boxed {\ 
 [X^a,|x\rangle\langle x|] \approx - i  \theta^{ab} \del_b(|x\rangle\langle x|)
\ . \ }
\label{X-comm-theta}
\end{align}
Furthermore, multiplying \eq{Xtheta-1}
with $(\del_c+iA_c)\langle x|$ gives
\begin{align}
 - i \theta^{ab} \big((\del_c+iA_c)\langle x|\big)(\del_b -iA_b)|x\rangle
&\approx -i\langle x|\cX_c (X^a -x^a + \del^a\l) |x\rangle
 =  -i\langle x|\cX_c X^a |x\rangle = P_{c}^{\ a}  \ ,
 \label{theta-P-1}
\end{align} 
and similarly from \eq{Xtheta-2}
\begin{align}
i \theta^{ac} ((\del_c + iA_c)\langle x|)(\del_b-iA_b)|x\rangle 
 &\approx  i\langle x|X^a \cX_b|x\rangle
 =  P^{a}_{\ b}  \ .
  \label{theta-P-2}
\end{align}
Adding these and using  \eq{projector-embed} and \eq{hab-def}  gives 
\begin{align}
\boxed{\ 
 - \theta^{ab}\omega_{bc} 
  \approx \del_c {\bf x^a} =  \del_c (x^a -\del^a\l) 
 \ }
 \label{theta-omega}
\end{align}
in the semi-classical regime, as in \cite{Ishiki:2016yjp}. The rhs will be recognized as 
tangential projector on the embedded quantum space 
$\tilde\cM \subset \R^D$.
Therefore the above relation states that 
$\theta^{ac}$ is tangential to $\tilde\cM$,
and the inverse of the would-be
symplectic form $\omega_{ab}$ on $\tilde\cM$. 
This implies that 
$\omega|_{\tilde\cM}$ is indeed non-degenerate i.e. symplectic, and
$\theta^{ac}$ is its associated Poisson structure\footnote{Recall that the Jacobi identity is a consequence of $d\omega = 0$.}.
Together with \eq{theta-def} we obtain
\begin{align}
 [X^a,X^b]|x\rangle \approx i\{x^a,x^b\}|x\rangle  
 = i \theta^{ab}|x\rangle  \
\end{align}
which can be written in the  notation of section \ref{sec:quatiz-semi} as 
semi-classical relation
\begin{align}
\boxed{\ 
[X^a,X^b] \sim i\{x^a,x^b\}  = i \theta^{ab}  \ . 
 \ }
 \label{XX-theta-rel}
\end{align}
Moreover, this means that {\bf almost-local quantum spaces $\cM$
can be locally approximated by some Moyal-Weyl quantum plane} $\R^{2n}_\theta$.
In particular, this implies
that {\bf the almost-K\"ahler condition} \eq{almost-Kahler} holds at least approximately.
Furthermore, taking the inner product of
\eq{Xtheta-1} and \eq{Xtheta-2} 
we obtain 
\begin{align}
(\Delta X^a)^2 =
 \langle x| (X^a -x^a + \del^a\l)(X^a -x^a + \del^a\l)|x\rangle 
%  &\sim \frac 14 \d_{aa'}\theta^{ab}  \theta^{a'c} 
%  ((\del_c + iA_c)\langle x|) (\del_b -iA_b)|x\rangle \nn\\
 &=  \theta^{ab}\theta^{ac} g_{bc} 
 \label{uncertainty-X}
\end{align}
(no sum over $a$), where $g_{bc}$ is the
quantum metric \eq{g-def}.
Hence the uncertainty of $X^a$ is 
characterized by the {\bf uncertainty length}\footnote{On 
 quantum K\"ahler manifolds, this reduces to the
well-known form $L_{\rm NC}^2 = \|\theta^{ab}\|$.} 
\begin{align}
 L_{\rm NC}^2  := \|\theta^{ab}\|^2 L_{\rm coh}^{-2} \ .
 \label{L-NC-def}
\end{align}
We also note the relation \cite{Ishiki:2016yjp}
 \begin{align}
  i \theta^{ac} g_{cb}  =  \d^{aa'}(P_{a'c} - P_{ca'})
  = 2i \d^{aa'}{\rm Im}(P_{a'c})
  \label{imaginary-P}
 \end{align}
which is obtained by subtracting \eq{theta-P-1} and \eq{theta-P-2};
in particular, $\theta^{ac} g_{cb}$ is antisymmetric.
 Finally, by comparing \eq{Xtheta-1} with \eq{cX-def}  we obtain
\begin{align}
 \theta^{ab} (\del_b -iA_b)|x\rangle \approx
 i(X^a -x^a + \del^a\l)|x\rangle
 =  (H_x-\l)i(\del^a-i A^a)|x\rangle \ ,
 \label{almost-Kahler}
\end{align}
which relates
$i(\del^a -iA^a)|x\rangle$ and $\theta^{ab}(\del_b -iA_b)|x\rangle$,
up to the action of $H_x-\l$. 
% As discussed  in section \ref{sec:Kahler}, 
% this could be interpreted as (almost-) K\"ahler geometry
% {\bf if} $H_x-\l$ acts trivially. However there is no obvious reason 
% why that should always be the case

\section{The abstract quantum space $\cM$}
\label{sec:abstract-quantum}

In the previous section we considered the bundle $\cB$ 
of quasi-coherent states  $|x\rangle$ over $\tilde\R^D$.
However, these states 
often coincide for different $x$. 
In this section we develop a general concept of quantum geometry which naturally
captures such situations, and  leads to a variety $\cM \subset \C P^{N-1}$,
which is naturally embedded in $\R^D$.

Consider the union of the normalized quasi-coherent 
states for all $x\in\tilde\R^D$ 
\begin{align}
\cB &:= \bigcup_{x\in\tilde\R^D} U(1)|x\rangle \  \subset \  \cH \cong \C^N 
\end{align}
as a subset of $\cH$; here the union need not be disjoint. 
$\cB$ can be viewed  as a $U(1)$ bundle\footnote{in slight abuse of notation we use the same letter $\cB$ 
as in section \ref{sec:quasicoherent}, hoping that no confusion arises.}
\begin{align}
 \cB \to \cM , \qquad \cM := \cB/_{U(1)} \ \hookrightarrow \C P^{N-1} \ 
 \label{M-sub-CP}
\end{align}
over $\cM$.
We denote $\cM$ as {\bf abstract quantum space associated to $X^a$}.
Thus $\cM$ inherits the induced 
(subset) topology and metric from $\C P^{N-1}$. 
A matrix configuration will be denoted as {\bf quantum manifold}
if $\cM\subset\C P^{N-1}$ is a regular (real) submanifold.
This is not far-fetched, since
standard theorems \cite{rellich1969perturbation,kato2013perturbation} ensure
the existence of (local) smooth maps 
\begin{align}
 {\bf q}: \quad U\subset\tilde\R^D &\to \cM\subset \C P^{N-1}  \nn\\
   x &\mapsto |x\rangle \ .
   \label{q-map}
\end{align}
Hence $\cM$ is ``locally translation invariant'',
with generators inherited from the  $SU(N)$ symmetry of $\C P^{N-1}$.
However, ${\bf q}$ need not be injective.
To understand this better, we note that
\begin{align}
 \cM \ \cong  \ \tilde\R^D/ _\sim \ 
 \label{cM-equivalence-def}
\end{align}
where the equivalence relation $\sim$ on 
$\tilde\R^D$ is defined by identifying points $x\in\tilde\R^D$
with identical eigenspace $E_x$.
Denote the equivalence class through a point $x\in\tilde\R^D$ with $\cN_x$. 
Due to the identity
\begin{align}
 H_x = H_y  + \frac 12(x^ax_a -y^a y_a)\one - (x^a-y^a)X_a\ ,
\end{align}
$x\sim y$ implies that $|x\rangle$ is an eigenvector of $(x^a-y^a)X_a$,
\begin{align}
 (x^a-y^a)X_a |x\rangle \propto |x\rangle \ .
 \label{x-y-Nx-eq}
\end{align}
But this means that the {\em equivalence classes $\cN_x$ are always (segments of) straight lines or higher-dimensional planes}\footnote{The $\cN_x$ either extend to infinity or 
end up at the singular set $\cK$, where the $|x\rangle$ may turn into higher eigenstates.}, and
it follows using \eq{deform-eigenstate} that 
\begin{align}
 w_a\cX^a|x\rangle = 0 = w_a(X^a - x^a + \del^a\l)|x\rangle, \qquad w\in T\cN_x
 \label{V-anihil}
\end{align}
along such directions. This implies via \eq{gab-XX} that  $\cN_x$
is a null space w.r.t. the quantum metric $g_{ab}$ induced from $\C P^{N-1}$.
The quantum metric hence characterizes the dependence of the coherent 
states along the non-trivial directions of $\cM$.
Moreover, kernel of $d{\bf q}$ at $x$ is given by $T\cN_x$.

The above observations provide a remarkable link between local and global 
properties of ${\bf q}$: {\em whenever ${\bf q}(x) = {\bf q}(y)$ for $x\neq y$, 
a linear kernel $T\cN_x \ni (x-y)$ of $d{\bf q}|_x$ arises}.
In particular if  ${\rm rank}\, d{\bf q}=D$ i.e.  ${\bf q}$ is an immersion, 
${\bf q}$ must be injective globally, since otherwise $d{\bf q}$
has some non-trivial kernel. 
This implies that ${\bf q}$ can be extended to $\tilde\R^D$, and

\begin{thm}
 If ${\bf q}$ \eq{q-map} is an immersion, then ${\bf q}:\, \tilde\R^D \to \cM$ is bijective, and
 $\cM$ is a $D$-dimensional quantum manifold. Moreover, $x^a$ provide global coordinates.
 
\end{thm}
An infinite-dimensional example is given by the Moyal-Weyl quantum plane, and the fuzzy disk \cite{Lizzi:2003ru} is expected to
provide a finite-dimensional example.
However, there are many interesting examples (such as the fuzzy sphere, see section \ref{sec:fuzzy-S2})
where the rank of $d{\bf q}$ is reduced. We can still make  non-trivial statements with some extra assumption: 

A quantum space $\cM$ %(or the underlying matrix configuration $X^a$) 
will be called  {\bf regular} 
if ${\rm rank}\, d{\bf q}=m$ is constant on $\tilde\R^D$. Then
the fibration $\tilde\R^D/_\sim$ is locally trivial, and
according to the rank theorem \cite{lee2013smooth}
we can choose  functions $y^\mu, \ \mu=1,...,m$ on a neighborhood of $\xi\in U \subset\tilde\R^D$ 
%(where $x^\mu$ can be a subset of the Cartesian $x^a$ coordinates after a suitable $SO(D)$ rotation)
such that the image ${\bf q}|_U \subset \cM\subset \C P^{N-1}$ is a submanifold of $\C P^{N-1}$. 
Since the only possible degeneracies of ${\bf q}$ are the linear fibers $\cN$, it follows that
%\footnote{A fully rigorous formulation of 
%the proof is omitted, hence this might strictly speaking be considered as a conjecture.}

\begin{thm}
For regular quantum spaces i.e.
 for ${\rm rank}\, d{\bf q}=m$  constant, 
  $\cM$ is a $m$-dimensional quantum manifold.
 
\end{thm}
In particular, there are no self-intersections of $\cM$, and $\tilde\R^D$ has the structure of a bundle over $\cM$. 
Clearly local versions of this statement can also be formulated;
e.g. if the rank of $d{\bf q}$ is reduced at some point, $\cM$ may be ``pinched''.
Furthermore, it may seem natural to conjecture that $\cM$
is compact, since $\cH$ is finite-dimensional; however, the 
proper statement should be that $\cM$ has a natural compactification:
since $H_x \to -x_a X^a$ for $|x|\to\infty$, the state $|x\rangle$
approaches the lowest eigenspace of $e_a X^a$ for $e = \frac{x}{|x|}\in S^{D-1}$.
Hence if $\cM$ does not already contain these 
states, then  $\cM$ could be compactified by adding them (and possibly other states).

Now consider the following natural {\em embedding map} provided by the symbol of $X^a$:
\begin{align}
\boxed{
\begin{aligned}
   {\bf x^a} :\quad \cM &\to \R^D  \\
    |x\rangle &\mapsto {\bf x^a}:= \langle x|X^a|x\rangle 
    = x^a - \del^a\l
\end{aligned}  
    }
    \label{M-embedding-symbol}
\end{align}
%(cf. (7.1) in \cite{Schneiderbauer:2016wub}) 
using \eq{expect-X}. This is  the quotient of the 
previously defined function ${\bf x^a}$ \eq{X-ecpect-embedding-symbol} on $\tilde\R^D$, which 
is constant on the fibers $\cN_x$.
The image 
\begin{align}
\boxed{\ 
 \tilde\cM := {\bf x}(\cM) \quad \subset \R^D
 \ }
 \label{embedded-M-def}
\end{align}
defines some variety in target space $\R^D$. 
In this way, we can associate  to the abstract space $\cM$ 
a subset $\tilde\cM \subset\R^D$, and  
$\cB$ can be considered as a $U(1)$ bundle over $\tilde\cM$.
This structure defines the
{\bf embedded quantum space} or {\bf brane} associated to 
the matrix configuration.
The concept is very reminiscent of noncommutative branes in string theory, which is borne out in the context of Yang-Mills matrix models, 
cf. \cite{Seiberg:1999vs,Aoki:1999vr,Steinacker:2008ri}. However the embedding might be degenerate, and the 
abstract quantum space is clearly a more fundamental concept.

If  equivalence class $\cN_x$ of  $x$ is non-trivial, further interesting statements can be made.
Observe that $\l(x) = \d^2(x) + d^2(x)$  reduces on $\cN_x$ to 
the displacement $d^2(x)$ plus a constant shift
$c = \d^2(x)$. Therefore there is a unique $x_0\in \cN_x$ in each 
equivalence class where $\l$ assumes its minimum.
This provides a natural representative of $\cM\cong \tilde\R^D/_\sim$, 
and another embedding function 
\begin{align}
 x_0^a: \quad \tilde\R^D \to \cM \hookrightarrow\R^D
\end{align}
which is constant on the fibers $\cN$ and faithfully represents\footnote{This also provides the natural adapted coordinates 
implied by the constant rank theorem \cite{lee2013smooth}.} $\cM$.
%In the semi-classical regime as discussed below, we expect that 
%${\bf x^a}(x_0) \approx x_0^a$. Indeed, 
It satisfies 
\begin{align}
 w_a({\bf x^a}(x_0) - x_0^a) = w_a\del^a\l|_{x_0} = 0
 \qquad \forall \ w\in T\cN_{x_0}
 \label{l-min-N}
\end{align}
using \eq{expect-X},
because $\l$ assumes its minimum on $\cN_{x_0}$ at $x_0$. 
Therefore ${\bf x^a}(x) = {\bf x^a}(x_0)$ provides the optimal estimator for $x_0$
in $\cN_x$, in the sense that 
\begin{align}
 x_0^a = P_x^\perp{\bf x^a}(x)
 \label{X-X0-proj}
\end{align}
 where 
$P_x^\perp$ is the orthogonal projector on $\cN_x$ w.r.t. the Euclidean metric on $\R^D$.
This provides justification for the numerical ``measuring algorithm'' in 
\cite{Schneiderbauer:2016wub,lukas_schneiderbauer_2016_45045}, and suggests further refinements.

\paragraph{Quantum tangent space.}

From now on, we will assume that $\cM$ is a quantum manifold.
Since $\cM\subset\C P^{N-1}$ is a (sub)manifold, we can determine its tangent space.
Choose some point $\xi\in\cM$.
The results of section \ref{sec:derivatives-generators} notably \eq{del-nabla-X-general-2} imply that $T_\xi\cM$ is spanned by
the $D$ vectors
\begin{align}
 (\del_a- i A_a)|x\rangle = i\cX_a|x\rangle \ \in T_\xi\C P^{N-1} \ ;
\end{align}
note that $\langle x|(\del_a - i A_a)|x\rangle=0$,
hence $i\cX_a|x\rangle$ is indeed a tangent vector\footnote{since $A_\mu$ can be gauged away at any given point, these are
derivatives of sections of the respective $U(1)$ bundles over $\cM$ and $\C P^{N-1}$, which can be taken as representatives 
of tangent vectors on $\cM$ and $\C P^{N-1}$, respectively.
Although the $\cX_a$ depend implicitly on $x$,
the result is independent of the point $x\in\cN_x$ because $\cM$ is a manifold.} 
of $\cM \subset \C P^{N-1}$, 
and perpendicular to the ``would-be vertical vector''  $i|x\rangle$.
According to \eq{V-anihil}, any $w\in T\cN_x$ provides a non-trivial relation $w^a\cX_a|x\rangle=0$. 
Hence after a suitable $SO(D)$ rotation, we can choose among the Cartesian coordinates on $\R^D$
$m$ local coordinates $x^\mu$ which are perpendicular\footnote{Since $\cN_x$ is in one-to-one correspondence with 
$\xi\in\C P^{N-1}$, we shall use this notation if appropriate.} 
to $\cN_\xi$, and can serve as local coordinates of $\cM$
near $\xi$. We denote these as local {\em ``normal embedding'' coordinates}  on $\cM$. 
It follows that an explicit basis of the tangent vectors in $T_\xi\cM$ is given by
$(\del_\mu- i A_\mu)|x\rangle = i\cX_\mu|x\rangle$ for $\mu=1,...,m$.
This provides a natural definition of 
the {\bf(real) quantum tangent space} of $\cM$:
\begin{align}
 T_\xi\cM = \Big\langle i\cX_\mu|x\rangle \Big\rangle_\R
 = \Big\langle i\cX_a|x\rangle \Big\rangle_\R \quad \subset \ \ T_\xi\C P^{N-1}
 \label{tangent-space-R}
\end{align}
with basis
$i\cX_\mu|x\rangle, \ \mu = 1,...,m$, so that $\dim T_\xi\cM = m = \dim\cM$.

%Here $T_x\C P^{N-1}$ has real dimension $2(N-1)$.
One can now repeat the considerations in section \ref{sec:inner-connect-symp}, 
in terms of local coordinates 
$x^\mu,\ \mu=1,...,m$ on $\cM$.
Thus $\cM$ is equipped with a $U(1)$ connection 
\begin{align}
 iA = \langle x|d|x\rangle
 \end{align}
 and a closed 2-form \eq{F-omega-dA}
\begin{align}
 i\omega_\cM = d\langle x|d|x\rangle = \frac i2\omega_{\mu\nu} dx^\mu \wedge dx^\mu = i dA, \qquad  d\omega_\cM = 0
\end{align}
as well as a quantum metric $g_{\mu\nu}$, which are simply the 
pull-back of the symplectic structure and the  
Fubini--Study metric on $\C P^{N-1}$. 
These structures are intrinsic, and have nothing to do with target space $\R^D$.
Given the basis $i\cX_\mu|x\rangle$ of tangent vectors, we can evaluate the
symplectic form and the quantum metric in local embedding coordinates as
\begin{align}
 i\omega_{\mu\nu} 
  &= \langle x|(\cX_\mu \cX_\nu - \cX_\nu \cX_\mu) |x\rangle \nn\\
   g_{\mu\nu} &= \langle x|(\cX_\mu \cX_\nu + \cX_\nu \cX_\mu) |x\rangle \ .
\end{align}
It should be noted that the quantum tangent space $T_x\cM$
of the abstract quantum space is a subspace of $\C P^{N-1}$, and 
has a priori nothing to do with the embedding in target space $\R^D$.
This is  indicated by the attribute ``quantum''.
The embedding \eq{M-embedding-symbol} in target space induces another 
metric on $\cM$, which in turn is distinct from the effective 
metric  discussed in section \ref{sec:eff-metric-MM}.

It is tempting to conjecture that 
 for irreducible matrix configuration, $\omega_\cM$ is always non-degenerate, 
 and thus defines a symplectic form on $\cM$.
However this is not true, as demonstrated by the minimal fuzzy torus or minimal fuzzy $H^4$
 where $\omega_\cM$ vanishes, cf. section \ref{sec:examples}. 
But if there is a  semi-classical regime, 
$\omega_\cM$ is indeed non-degenerate  and thereby a symplectic manifold, as discussed in the next section\footnote{
For reducible matrix configuration $\omega_\cM$ may be degenerate even in the semi-classical regime.}. 
From now on we will mostly drop the subscript from $\omega_\cM = \omega$.

\paragraph{Embedded quantum space for almost-local quantum spaces.}

Now consider the tangent space $T\tilde\cM$  of the embedded brane 
$\tilde\cM \subset\R^D$ \eq{embedded-M-def}, which is spanned by $\del_\mu {\bf x^a}$
for any local coordinates on $\cM$.
This can be understood for almost-local quantum spaces,
following the semi-classical analysis of \cite{Ishiki:2016yjp}.
Recall the relation \eq{theta-omega} 
\begin{align} 
  \frac{\del}{\del x^c} {\bf x^a} \ \approx - \theta^{ab}\omega_{bc} 
 % = \d^{ab} (P_{bc} + P_{cb}) \ .
\end{align}
as tensors on $\tilde\R^D$. It follows that $\theta^{ab}$ is non-degenerate on 
$\tilde\cM$. Then $0 \approx  i\theta^{ab} \del_b\l$ \eq{theta-del-l} 
implies that $\l$ is approximately constant on 
$\tilde\cM$,
and the derivative of $\l$ along the transversal fiber $\cN$ (approximately) vanish
on $\tilde\cM$ due to \eq{l-min-N}. Then \eq{expect-X} implies 
\begin{align}
{\bf x^a}(x) \approx x^a, \qquad
 \ \del_\mu {\bf x^a} \approx \del_\mu x^a 
  \label{embed-x-approx}
\end{align}
so that both tensors $\theta^{ab}$ and $\omega_{bc}$ are approximately tangential to $\tilde\cM$,
and inverse of each other on $\tilde\cM$. This is particularly transparent in normal embedding coordinates.
In particular, $\tilde\cM$ is the location where $\l$ assumes its ``approximate'' minimum, 
which was used in \cite{Schneiderbauer:2016wub,lukas_schneiderbauer_2016_45045} to numerically measure and picture such branes.
Then the embedding map \eq{M-embedding-symbol} is an immersion, 
but (the closure of)  $\tilde\cM \subset\R^D$ 
may have self-intersections,
as in the example of squashed $\C P^2$ \cite{Steinacker:2014lma}.
Both $\omega_{ab}$ and $g_{ab}$ vanish along the directions $w^a$ along the fiber $\cN$,
\begin{align}
 w^a\omega_{ab} = 0 = w^a g_{ab} \ , \qquad  w\in T \cN \ .
\end{align}
Finally, we can recognize \eq{projector-embed}
\begin{align}
 \del^a {\bf x^a}  = P^{ab} + P^{ba} 
\label{P-M-proj}
\end{align}
as tangential projector on $\tilde\cM\subset\R^D$, since 
the rhs vanishes along the fibers $\cN$.
This was obtained in \cite{Ishiki:2016yjp} in the semi-classical limit,
but that relation holds in fact exactly.

\subsection{Quantization map, symbol and semi-classical regime}
\label{sec:quatiz-semi}

Given the quasi-coherent states, we can define a {\bf quantization map}
 \begin{align}
  \cQ: \quad \cC(\cM) &\to End(\cH) \nn\\
   \phi(x) &\mapsto \int_\cM \phi(x) \,|x\rangle \langle x|
   \label{Q-map}
 \end{align}
which associates  to every classical 
 function on $\cM$ an operator or observable in $End(\cH)$.
 The integral on the rhs is defined\footnote{This is well-defined if 
 (the closure of) $\cM$ is a compact sub-manifold of $\C P^{N-1}$, which we shall assume.
 It is essential to use the abstract quantum space $\cM$ here, 
 otherwise the integral would typically not make sense.}  naturally via the symplectic volume form
 \begin{align}
 \int_\cM \phi(x) := \frac{1}{(2\pi\a)^n} \int\limits_\cM \Omega \, \phi(x) \ ,
 \qquad \Omega :=  \frac{1}{n!}\omega^{\wedge n} 
 \label{int-def}
 \end{align}
 (assuming  $\dim\cM = m=2n$),
where the normalization factor $\a$  is defined  by
\begin{align}
 N = Tr(\one) = \int_\cM 1 \ .
 \label{symp-vol}
\end{align} 
Semi-classical considerations suggest that $\a\approx 1$,
however this cannot hold in full generality, since the symplectic form is degenerate for the minimal fuzzy torus and the integral vanishes.
It would be desirable to find sufficient conditions for $\a\approx 1$, 
and a precise statement in particular for
the quantum K\"ahler manifolds discussed below. In any case, the trace is related to the intragel via
\begin{align}
 Tr\cQ(\phi) = \int_\cM \phi(x) \ .
 \label{Tr-Q}
\end{align}
 The map $\cQ$  cannot be injective, since $End(\cH)$ is finite-dimensional;
 the kernel is typically given by functions with high ``energy''.
 It is not evident in general if this map is surjective, which
 will be established below for the case of quantum K\"ahler manifolds.
 
 We can now  re-define the  {\bf symbol map} \eq{symbol-RD} more succinctly as
\begin{align}
 End(\cH) &\to \cC(\cM)   \nn\\
  \Phi &\mapsto \langle x|\Phi|x\rangle =: \phi(x) \ .
  \label{symbol}
\end{align}
Both sides have a natural norm and inner product, given by
\begin{align}
 \langle \Phi,\Psi\rangle = Tr(\Phi^\dagger\Psi)\quad \ \mbox{and}\quad \  
 \langle\phi,\psi\rangle = \int_\cM \phi(x)^*\psi(x)
\end{align}
leading to the Hilbert-Schmidt norm $\|\Phi\|_{HS}$ and the $L^2$ norm $\|\phi\|_2$, respectively. 
The symbol map can be viewed as de-quantization map, which makes sense for any quantum space 
in the present framework. 

The concept of almost-local operators discussed in section \ref{sec:almost-local} can now also be refined.
We re-define $Loc(\cH)\subset End(\cH)$ as a maximal (vector) space of 
operators such that the restricted symbol map 
\begin{align}
   Loc(\cH) &\to\cC_{\rm IR}(\cM) \nn\\
  \Phi \ &\mapsto  \langle x|\Phi|x\rangle =: \phi(x) 
  \label{symbol-M}
\end{align}
is an ``approximate isometry''
with respect to the Hilbert-Schmidt norm on $Loc(\cH) \subset End(\cH)$
and the $L^2$-norm on  $\cC_{\rm IR}(\cM)\subset L^2(\cM)$. 
We will then identify $\Phi \sim \phi$.
Approximate isometry means  that $|\|\phi\|_2 - 1| < \varepsilon$ whenever $\|\Phi\|_{\rm HS}=1$
for some given $0 < \varepsilon < \frac 12$, depending on the context.
% Note that 
% $Loc(\cH)$ is closed under scalar multiplication, but is not automatically a vector space; however if $\Phi,\Psi$ and $\Phi+\Psi$ are in $Loc(\cH)$, then the 
Then the polarization identity implies
\begin{align}
 \langle\Phi,\Psi\rangle_{\rm HS} \approx \langle\phi,\psi\rangle_2 \ ,
\end{align}
%One may also require that $Loc(\cH)$ is a vector space; 
hence an ON basis of $Loc(\cH)$ is mapped to a basis of $\cC_{\rm IR}(\cM)$
which is almost ON.
This defines the {\bf semi-classical regime}, which
can be made more precise in some given situation by specifying 
some $\varepsilon$. 
Accordingly, {\bf almost-local quantum spaces} are (re)defined as matrix configurations  
where all $X^a$ and $[X^a,X^b]$ are in $Loc(\cH)$.

Of course some given matrix configuration may be far from any semi-classical space, 
in which case $Loc(\cH)$ is trivial. 
However we will see that for almost-local quantum space, $Loc(\cH)$ 
typically includes the almost-local operators 
in the sense of \eq{Phi-loc-approx} up to some bound, and in particular polynomials in $X^a$
up to some order. Moreover, $\cQ$ is an approximate inverse of 
the symbol map \eq{symbol-M} on $Loc(\cH)$. Then the semi-classical regime
should contain a sufficiently large class 
of functions and operators to characterize the geometry to a satisfactory precision.

Let us try to justify these claims.
The first observation is that $\one \in  Loc(\cH)$, because its symbol is the constant function 
$1_\cM$, and the norm is preserved due to \eq{symp-vol}.
Conversely, we should show the {\em completeness relation} 
\begin{align}
\cQ(1_\cM) =  \int_\cM |x\rangle\langle x| \
\stackrel{!}{\approx} \ \one \ 
\label{one-approx}
\end{align}
which is equivalent\footnote{The following considerations would also go through if these relations 
hold with some non-trivial density.} to the trace identity
\begin{align}
 Tr\Phi = \int_\cM \langle x|\Phi|x\rangle \qquad \forall \Phi\in End(\cH) \ .
\end{align}
This is not automatic, since the integral vanishes e.g. on minimal $T^2_2$.
We can establish the completeness relation 
at least formally\footnote{A more precise statement \eq{one-H0-coherent} 
will be shown for
quantum K\"ahler manifold.}
(or rather approximately) for almost-local quantum spaces.
%as defined in section \ref{sec:almost-loc-space}. 
Indeed then \eq{X-comm-theta} implies
\begin{align}
 [X^a,\cQ(\phi)] &\approx -i\int_\cM \phi(x) \theta^{ab} \del_b(|x\rangle\langle x|)  \nn\\
  &= i\int_\cM \theta^{ab} \del_b\phi(x) |x\rangle\langle x| \nn\\
  &= \cQ(i\theta^{ab} \del_b\phi)
  \label{X-comm-Q}
\end{align}
because the integration measure $\Omega$ \eq{int-def} is invariant under 
Hamiltonian vector fields.
In particular, $\cQ(1_\cM)$ (approximately) commutes with all $X^a$, which by irreducibility
implies  $\cQ(1_\cM) \propto \one$, and \eq{one-approx} follows using 
the trace \eq{Tr-Q}.

Now assume that the completeness relation  holds to a sufficient precision.
Let $\Phi$ be an almost-local hermitian operator as defined in section \ref{sec:almost-local},
with symbol $\phi$.
% Then the completeness relation gives
% \begin{align}
%  \Psi  &= \int_\cM\, \Psi|x\rangle \langle x|
%   \approx \int_\cM \psi(x)\, |x\rangle \langle x| 
%    = \cQ(\psi)  \nn\\
%   \Phi \cQ(\psi) &= \int_\cM \psi(x) \, \Phi|x\rangle \langle x|
%   \approx \int_\cM \psi(x) \phi(x)\, |x\rangle \langle x| = \cQ(\psi\phi) 
%   \label{Q-Psi-Phi}
% \end{align}
% using \eq{Phi-loc-approx}; note that the symbols of almost-local operators
% are approximately constant on the coherence scale. Assuming that 
% $\Psi^\dagger$ also satisfies \eq{Phi-loc-approx},
% it follows that 
% \begin{align}
%   \Psi^\dagger \Psi  &\approx \cQ(\psi^\dagger\psi)  \nn\\
%  Tr(\Psi^\dagger \Psi) &\approx Tr\cQ(\psi^\dagger\psi) 
%  = \int\psi^\dagger\psi = \|\psi\|_2^2  \  
% \end{align}
% hence $\sim$ is indeed an isometry.
Then the trace relation gives
\begin{align}
  \|\Phi\|_{\rm HS}^2 &\approx \int_\cM\langle x|\Phi\Phi|x\rangle \approx 
 \int_\cM \phi(x)^2 = \|\phi\|_2^2
\end{align}
using  \eq{Phi-loc-approx}.
Therefore  almost-local 
operators in the sense of \eq{Phi-loc-approx} are indeed contained in $Loc(\cH)$,
%and in the  image of $\cQ$, 
up to the specific bounds. 
Conversely, assume that $\|\Phi\|_{\rm HS} \approx \|\phi\|_2$ for hermitian $\Phi$. 
Then the completenes relation implies 
\begin{align}
 \|\Phi\|_{\rm HS}^2 \approx \int_\cM \langle x|\Phi\Phi|x\rangle 
 &\approx  \int_\cM \phi(x)^2 = \|\phi\|_2^2  \nn\\
 \int_\cM \langle x|(\Phi - \phi(x))(\Phi - \phi(x))|x\rangle
  &\approx 0
  \label{norm-uncertainty}
\end{align}
which implies that $(\Phi-\phi(x))|x\rangle \approx 0$
$\forall x\in\cM$. Hence they are approximately local in the sense of \eq{Phi-loc-approx}. In particular they  
approximately commute due to \eq{semiclass-fact},
\begin{align}
 \Phi \Psi \approx \Psi \Phi, \qquad \Phi, \Psi \in Loc(\cH) \ .
\end{align}
 Hence the above definition of $Loc(\cH)$ is a refinement of the 
definitions in section \ref{sec:almost-local},
turning the local statements into global ones.

The  image $\cC_{\rm IR}(\cM)$ is typically given by functions which are  slowly varying  on the 
length scale  $L_{\rm coh}$, corresponding to the
semi-classical or infrared regime.
To see that $\cQ$ is approximately inverse to the symbol map, 
we note that the completeness relation  implies 
\begin{align}
 |y\rangle &\approx \int_\cM |x\rangle\langle x|y\rangle \ .
\end{align} 
 This means that 
\begin{align}
 \langle x|y\rangle \approx \d_y(x)
\end{align}
for any $y\in\cM$ w.r.t. the  measure \eq{int-def}, 
consistent with   
$|\langle x|y\rangle| \sim e^{-\frac 12(x-y)_g^2}$ \eq{coherent-inner} \eq{D-expand}. Then 
 \begin{align}
 \cQ(\phi)|y\rangle &\approx \int_\cM \phi(x)|x\rangle\langle x|y\rangle 
\approx \phi(y)|y\rangle .
\label{Q-phi-loc}
\end{align}
for functions $\phi(x)$
which are slowly varying on $L_{\rm coh}$.
Therefore $\cQ(\phi)$ is almost-local and hence 
$\cQ(\phi)\in Loc(\cH)$ for slowly varying $\phi$, 
and moreover 
$\cQ$ is approximately the inverse of the symbol map on $Loc(\cH)$, since
\eq{Q-phi-loc} gives
\begin{align}
 \langle y|\cQ(\phi)|y\rangle &\approx \phi(y) \ .
\end{align}
%which almost commute
%  \begin{align}
%   \cQ(f)\cQ(g) &\approx \cQ(fg) \approx  \cQ(g)\cQ(f) \ .
%  \end{align}
For almost-local quantum spaces, $Loc(\cH)$ contains in particular 
the basic matrices
\begin{align}
 X^a \approx \int_{\cM} {\bf x^a} |x\rangle\langle x| \ .
\end{align}
The approximation is good as long as the classical function 
${\bf x^a}$ is approximately constant on $L_{\rm coh}$.
Moreover,
\eq{XX-theta-rel} gives the approximate commutation relations on $\cM$ 
\begin{align}
 [X^a,X^b] \sim i \theta^{ab} = i\{x^a,x^b\} \ .
\end{align}
We have seen that $\theta^{ab}$ is tangential to 
$\cM$ and the inverse of the symplectic form $\omega$ on $\cM$, hence 
$\{x^a,x^b\}$ are  Poisson brackets on $\cM$. 
In this sense, the semi-classical geometry is encoded in the matrix configuration $X^a$.
These observations are summarized in table \ref{tab:correspondence}.
\begin{table}[h]
\begin{center}
 \begin{tabular}{c c c}
  $Loc(\cH) \subset End(\cH)$ & $\sim$ &  $\cC_{\rm IR}(\cM) \subset L^2(\cM)$  \\  \hline 
$ \Phi  $ & $\sim$ & $ \phi(x) = \langle x|\Phi|x\rangle$  \\[1ex] 

$ X^a  $ & $\sim$ & $ {\bf x^a}(x) $  \\[1ex] 

$  [.,.] $ & $\sim$ & $ i\{.,.\} $  \\[1ex]

$  Tr $ & $\sim$ & $ \int_\cM $  \\[1ex] 
$\Box$ & $\sim$ &  $e^{-\sigma}\Box_G$  \\   
 \end{tabular}
\caption{Correspondence between almost-local operators 
and infrared functions on $\cM$ for almost-local quantum spaces. The  metric structure is encoded 
in the Laplacian $\Box$ \eq{Box-G}.}
\label{tab:correspondence}
\end{center}
\end{table}
 This provides the starting point of the emergent geometry and gravity
considerations in \cite{Steinacker:2010rh,Steinacker:2019fcb}, 
which will be briefly 
discussed in section \ref{sec:eff-metric-MM}.

The above Poisson structure extends trivially to 
$\tilde \R^D$, which for $D>\dim\cM$
decomposes into symplectic leaves of $\omega_{ab}$ that are preserved by the Poisson structure. Functions which are constant
on these leaves then have vanishing Poisson brackets, which leads to a degenerate effective metric 
as discussed in section \ref{sec:eff-metric-MM}.

 In the UV or deep quantum regime, 
 the above semi-classical picture is no longer justified, and in fact 
 it is very misleading.
 In particular, consider {\em string states} which are defined as 
 rank one operators built out of  quasi-coherent states 
\cite{Steinacker:2016nsc,Iso:2000ew}
\begin{align}
\psi_{x,y}
 &:= |x\rangle\langle y | \qquad \in End(\cH)  \ .
 \label{string-states}
\end{align}
They are highly non-local for $x\neq y$, and should not be interpreted as function
but rather as open strings (or dipoles) linking $|y\rangle$ to $|x\rangle$ 
on the embedded brane $\tilde\cM$.
These states provide a complete and more adequate picture of $End(\cH)$,
and  exhibit the stringy nature of noncommutative field theory 
and Yang-Mills matrix models \cite{Steinacker:2016nsc}. 
%%HS V2
This means that the physical content of Yang-Mills matrix models, 
and more generally of noncommutative field theory, 
is much richer than suggested by the semi-classical limit.
In particular, string states arise as high-energy excitation modes,  leading to 
UV/IR mixing in noncommutative field theory \cite{Minwalla:1999px}.
This is a phenomenon which has no counterpart in
conventional (quantum) field theory.

\subsection{Complex tangent space and quantum K\"ahler manifolds}
\label{sec:Kahler}

Now we return to the exact analysis. For any quantum manifold
$\cM$, the  embedding   $\cM \to \C P^{N-1}$
induces the tangential map
\begin{align}
 T_\xi\cM &\to T_\xi\C P^{N-1} \ .
  \label{M-embed-CP-d}
\end{align}  
Now we take into account that $\C P^{N-1}$  carries an intrinsic complex structure
\begin{align}
  \cJ:\quad  T_\xi\C P^{N-1} &\to  T_\xi\C P^{N-1}, \qquad  \cJ v = i v
 \label{complex-intrinsic}
\end{align}
 for any $v\in T_\xi\C P^{N-1}$. Accordingly,
$T\C P^{N-1} \cong T^{(1,0)}\C P^{N-1}$ can be viewed as holomorphic
tangent bundle, thus bypassing an explicit complexification of its real tangent space.
With this in mind, we  define the {\bf complex quantum tangent space} of $\cM$ as
\begin{align}
 T_{\xi,\C}\cM :=  \Big\langle\cX_a|x\rangle \Big\rangle_\C \quad \subset \ \ T_\xi\C P^{N-1} \cong  T_{\xi,\C}\C P^{N-1} \ ,
 \label{tangent-space-C}
\end{align}
 which also carries the  complex structure 
\begin{align}
 \cJ \cX_a|x\rangle := i\cX_a|x\rangle \quad \in T_{\xi,\C}\cM \ , \qquad 
 \cJ^2 = -\one \ .
\end{align}
Again, this complex tangent space 
is not necessarily the complexification of the real one.
Using the basis $i\cX_\mu|x\rangle, \ \mu = 1,...,m$ of $T_{\xi}\cM$
which arises in normal embedding coordinates, there may be relations 
of the form
\begin{align}
 (i \cX_\mu - J_\mu^{\ \nu}\cX_\nu)|x\rangle = 0 \quad \mbox{for} \quad
 J_\mu^{\ \nu} \in\R \ ,
 \label{complex-tangent-rel}
\end{align}
so that $T_{\xi,\C}\cM$ has reduced dimension over $\C$.
We will see that for quantum K\"ahler manifolds as defined below,
the complex dimension is half of the same as the real one.

\paragraph{Quantum K\"ahler manifolds.}

Consider the maximally degenerate case where the complex dimension of
 $T_{\xi,\C}\cM$ is given by $n= \frac m2 \in\N$ where $m= \dim_\R\cM$.
Then  $T_{\xi}\cM$ is stable under the  complex
structure operator $\cJ$ 
 \begin{align}
  T_{\xi,\C}\cM = T_{\xi}\cM
  \label{Kahler-cond}
 \end{align}
so that $T_\xi\cM$ should be viewed as holomorphic tangent space of $\cM$.
But this implies that $\cM\subset \C P^{N-1}$ is a complex sub-manifold (i.e. defined by holomorphic equations), 
cf. \cite{voisin2003hodge} or Proposition 1.3.14 in \cite{Baouendi:1999uya}.
% see also e.g.   Lemma 10 in 
% \begin{verbatim*}
% https://complex.univie.ac.at/fileadmin/user_upload/p_complex_analysis/skriptenlamel/CRGeometry_2013.pdf
% \end{verbatim*} .
Such quantum manifolds $\cM$ will be called {\bf quantum  K\"ahler manifolds}, 
for reasons explained below.
Indeed, all complex sub-manifolds  of $\C P^{N-1}$ are known to be K\"ahler.
 Note that  this is
an intrinsic property of a quantum space $\cM$, and no extra structure is introduced here: 
$\cM$ either is or is not of this type\footnote{
It is interesting to note that due to \eq{almost-Kahler},
$H_x$  preserves the complex tangent space $T_{\xi,\C}\cM$,
at least in the semi-classical regime. However, \eq{almost-Kahler} is still weaker than the K\"ahler condition.}.
We will see that this includes the well-known quantized or ``fuzzy'' spaces  
arising from quantized coadjoint orbits\footnote{It is worth pointing out that that $\C P^{N-1}$ is itself a quantum K\"ahler manifold,
as  minimal fuzzy $\C P^{N-1}_N$.}.

Consider the quantum K\"ahler case in more detail.
We can  introduce a local holomorphic parametrization of $\cM\subset\C P^{N-1}$ near $\xi$ 
 in terms of  $z^k\in\C^n$.
 Then any  local (!) holomorphic section of the tautological line bundle over 
$\C P^{N-1}$ defines via pull-back a local holomorphic section 
 of the line bundle
\begin{align}
 \tilde\cB := \bigcup_{x\in\tilde\R^D} E_x \to \cM \ \hookrightarrow \C P^{N-1}
\end{align}
over $\cM$, denoted by $\|z\rangle$. 
 This $\|z\rangle$ can be viewed as holomorphic $\C^N$-valued function
on $\cM$, which satisfies
\begin{align}
 \frac{\del}{\del \bar z^k} \|z\rangle = 0, \qquad  \|z\rangle\big|_\xi = |\xi\rangle
 \label{holo-anihil-coh}
\end{align} 
where $\bar z^k$ denotes the complex conjugate of $z^k$.
 Hence $\|z\rangle$  arises from $|x\rangle$ 
 through a re-parametrization and gauge transformation along with a 
 non-trivial normalization\footnote{$\|z\rangle$
cannot be normalized, since e.g. $\langle y\|z\rangle$ must be holomorphic in $z$.
Apart from that, $\tilde\cB$ is equivalent to $\cB$.} factor;
 this is indicated by the double line in $\|z\rangle$.
 In other words, the differential of the section
\begin{align}
 d\|z\rangle = dz^k \frac{\del}{\del z^k}\|z\rangle \qquad \in\ \Omega^{(1,0)}_{z}\cM 
 \label{holo-diff-sect}
\end{align}
is a  $(1,0)$ one-form.
Given this holomorphic one-form $d\|z\rangle$ and the hermitian inner product on $\cH$,
we naturally obtain a $(1,1)$ form
\begin{align}
\omega := (d\|z\rangle)^\dagger \wedge d\|z\rangle 
 &= \omega_{\bar k l} d\bar z^k \wedge dz^l \qquad \in \ \Omega^{(1,1)}_z\cM \nn\\
 \omega_{\bar k l} &= (d_k\|z\rangle)^\dagger d_l\|z\rangle 
\end{align}
which is closed,
\begin{align}
 d\omega = -(d\|z\rangle)^\dagger \wedge d d\|z\rangle + (d d\|z\rangle)^\dagger \wedge d\|z\rangle = 0
\end{align}
using holomorphicity of $\|z\rangle$.
This is the  K\"ahler form, which encodes the 
$\omega_{ab}$ in \eq{omega-def}. 
%% changed V2
As in \eq{hab-def},
we can then define  the hermitian metric
\begin{align}
 h(X,Y) &=   \big((d\|z\rangle)^\dagger \otimes d\|z\rangle \big)(X,Y)
  \qquad \in T^{(1,1)}
  \label{h-def-general}
\end{align}
whose imaginary and real part define the symplectic form and the quantum metric via
\begin{align}
  \omega(X,Y) &= -i(h(X,Y) - h(Y,X)^*) = -  \omega(Y,X)  \nn\\
  g(X,Y) &= h(X,Y) + h(X,Y)^* = g(Y,X) \ .
\end{align}
Since $h\in T^{(1,1)}$, they satisfy the compatibility condition
\begin{align}
 \omega(X,\cJ Y) &=  -i(h(X,\cJ Y) - h(\cJ Y,X)^*) \nn\\
   &= -i(ih(X,Y) +i h(Y,X)^*)  \nn\\
    &= g(X,Y)
    \label{Kahler-condition}
\end{align}
 (recall that $\cJ=-i$ on anti-holomorphic $(0,1)$ forms).
This means that $\cM$ is a K\"ahler manifold,
and the name ``quantum K\"ahler manifold''
indicates its origin from the matrices $X^a$.
% The quantum  metric  then satisfies
% \begin{align}
%  2 g(X,Y) &=   \big((d\|z\rangle)^\dagger \otimes d\|z\rangle \big)(X,Y) + (X\leftrightarrow Y)  \nn\\
%  &= \omega(X,\cJ Y) + \omega(Y,\cJ X)
%   = 2\omega(X,\cJ Y) \qquad \in T^{(1,1)}
% \end{align}
% which justifies the name ``quantum K\"ahler manifold''.
% Note that the metric is manifestly in $T^{(1,1)}$,
% hence $g(\cJ X,\cJ Y) = g(X,Y)$ is automatic.
In particular, the coherence length $L_{\rm coh}$ and the uncertainty scale 
$L_{NC}$ coincide.

Now we relate this to the local generators $\cX_\mu$ \eq{del-nabla-X-general-2}, \eq{tangent-space-R}.
Introducing real coordinates $z^k = z^k(x^\mu)$ where   % and $x^\mu = (x^k,y^k)$, 
$x^\mu$ are the local (Cartesian) embedding coordinates introduced above, 
the holomorphicity relation \eq{holo-anihil-coh} can  be expressed using \eq{del-nabla-X-general-2} as
\begin{align}
 0 = \frac{\del}{\del \bar z^k} \|z\rangle 
 = \frac{\del x^\mu}{\del \bar z^k} \frac{\del}{\del x^\mu} \|z\rangle 
 = i\frac{\del x^\mu}{\del \bar z^k}\, (\cX_\mu+A_\mu) \|z\rangle \ .
\end{align}
Similarly,
\begin{align} 
  \frac{\del}{\del z^k} \|z\rangle
  = \frac{\del x^\mu}{\del z^k} \frac{\del}{\del x^\mu} \|z\rangle 
  = i\frac{\del x^\mu}{\del z^k}(\cX_\mu+A_\mu) \|z\rangle \ .
\end{align}
We can now introduce new generators\footnote{The $\cA^k, \bar \cA_l$ are matrix-valued functions on $\cM$
just like the $\cX_\mu$, while the $X_a$  are ``constant'' matrices.} 
$\cA^k, \bar \cA_l$ via 
\begin{align}
 \cA^k &= i\frac{\del x^\mu}{\del \bar z^k}\,(\cX_\mu+A_\mu) \nn\\
 \bar\cA_k &= i\frac{\del x^\mu}{\del z^k}\, (\cX_\mu+A_\mu)
\end{align}
so that 
\begin{align}
 \cA^k \|z\rangle  &= 0,   \qquad 
 \bar\cA_k \|z\rangle  =  \frac{\del}{\del z^k} \|z\rangle \ .
\end{align}
These are clearly the analogs of the standard annihilation properties 
of coherent states.
It is hence appropriate to denote the $\|z\rangle$ on
quantum K\"ahler manifolds as {\bf coherent states}.
Then
\begin{align}
 T_{\xi,\C}\cM = \Big\langle \bar\cA_k \|z\rangle \Big\rangle_{\C} \ \cong \C^n \, \qquad k=1,...,n .
\end{align}
The metric tensor and the symplectic form are then determined as usual by the K\"ahler form
\begin{align}
 i\omega_{\bar k l} &= (d_k\|z\rangle)^\dagger d_l\|z\rangle 
 = \langle z\| \bar\cA_k^\dagger \bar\cA_l\|z\rangle 
 % = \langle z\| [\bar\cA_k^\dagger, \bar\cA_l]\|z\rangle 
\end{align}
which arises from a local K\"ahler potential, 
\begin{align}
 \omega_{\bar k l} &= -\frac 12 \bar\del_k\del_l \rho
\end{align}
given by the restriction of the 
(Fubini--Study) K\"ahler potential  on $\C P^N$.

This provides a rather satisfactory concept of quantum K\"ahler geometry, which arises in a natural way from the complex structure in the 
Hilbert space.
There is no need to invoke any semi-classical or large $N$ limit.
Not all quantum spaces are of this type,  
a counterexample being the minimal fuzzy torus $T^2_2$ 
as discussed in section \ref{sec:min-fuzzy-T2}. 
In \cite{Ishiki:2016yjp}, it is claimed that all quantum manifolds are K\"ahler in the 
semi-classical limit, based on \eq{imaginary-P}. However this refers to a
different almost-complex structure and metric which is not intrinsic.
From the present analysis, there is no obvious reason why all 
quantum manifolds should be K\"ahler, even in the semi-classical limit.

Since for non-K\"ahler manifolds the 
complex tangent space $T_{\C}\cM$ is higher-dimensional, 
quantum effects due to loops in Yang-Mills matrix models
may be more significant, and the geometric trace formula 
(2.38) in \cite{Steinacker:2016nsc} for string states would  need to be replaced with some 
higher-dimensional analog.
This suggests that quantum K\"ahler manifolds may be protected 
by some sort of non-renormalization theorems.

\section{Coherent states and quantization map for quantum K\"ahler manifolds}
\label{sec:coherent}

We can establish the following  lemma, which is well-known for
 standard coherent states:

\begin{lem}
 \label{lemma-diagonal}
 Let $|x\rangle$ be the coherent states of a quantum K\"ahler manifold $\cM$,
 and $\cH_0\subset\cH$ their linear span.
 Assume $A\in End(\cH_0)$ satisfies $\langle x|A|x\rangle = 0$ for all $x\in\cM$.
 Then $A=0$.
 
\end{lem}

\begin{proof}
 
 Consider the function
 \begin{align}
  A(\bar y,z) := \langle y\|A\|z\rangle
 \end{align}
 where $\|z\rangle, \|y\rangle$ are local holomorphic sections of the coherent states in a neighborhood of $\xi\in\cM$.
Clearly this function is holomorphic in $z$ and in $\bar y$.
 By assumption, the restriction of $A(\bar y,z)$ to the diagonal 
 $A(\bar z,z) = \langle z\|A\|z\rangle$ vanishes
 identically. But then the standard properties of holomorphic 
 functions imply (cf. \cite{Perelomov:1986tf}) that $A(\bar y,z)\equiv 0$ identically. 
 This argument applies 
 near any given point $\xi\in\cM$, which  implies that  $A =0$.

\end{proof}

Using this lemma, we can establish the diagonal realization of 
operators via coherent states:

\begin{thm}
\label{thm:diag-coh}
 
 Let $|x\rangle$ be the (normalized)
 coherent states of a quantum K\"ahler manifold $\cM$,
 and $\cH_0\subset\cH$ their linear span.
 Then all operators $A\in End(\cH_0)$ can be written as 
 \begin{align}
   A =  \int_\cM A(x) \,|x\rangle \langle x|
  \label{A-diag-rep}
 \end{align}
 for some suitable complex-valued function $A(x)$ on $\cM$.
 
\end{thm}
 
 Note that if the holomorphic coherent states $\|x\rangle$ are used instead of the normalized $|x\rangle$, 
 then $A(x)$ might have some singularities.

\begin{proof}
 
 Assume that the subspace in $End(\cH_0)$ spanned by the rhs of \eq{A-diag-rep} 
 is smaller than $End(\cH_0)$. Let $B\in End(\cH_0)$ be in its orthogonal 
 complement w.r.t. the Hilbert-Schmidt metric.
 Then
 \begin{align}
  0 = Tr( A B) = \int_\cM  A(x) \langle x|B|x\rangle \qquad \forall A(x)\in\cC(\cM). 
 \end{align}
But this implies $\langle x|B|x\rangle = 0 \ \forall x\in\cM$, and then by Lemma 
\ref{lemma-diagonal} it follows that $B=0$.
  
 \end{proof}
 
Consider again the span
$\cH_0\subset \cH$ of all quasi-coherent states $|x\rangle$.
It is natural to conjecture  
\begin{conj}
 For every irreducible matrix configuration, $\cM$ is connected, and
 the quasi-coherent states
 are over-complete, i.e.
 \begin{align}
  \cH_0 = \Big\langle |x\rangle ; x\in\tilde\R^D\Big\rangle_\C  = \cH \ .
 \end{align}
 
\end{conj}

In the semi-classical regime this follows from \eq{one-approx}
and \eq{X-comm-Q}, which would give a central element for every connected 
component of $\cM$.
% Alternatively, it follows from the observation 
% $T\cM \subset \cH_0$ together with \eq{Xtheta-1} which implies  
% $X^a|x\rangle\in \cH_0$.
A viable general strategy to show this more generally might be to show that the 
continuation of the $|x\rangle$ through the singular set $\cK$
provides all eigenstates of $H_x$. 
However, this is left as a conjecture.

In any case, we can consider the following restricted form of the
quantization map  \eq{Q-map} 
 \begin{align}
  \cQ: \quad \cC(\cM) &\to End(\cH_0) \nn\\
   \phi(x) &\mapsto  \int_\cM \phi(x) \,|x\rangle \langle x|
 \end{align}
 associating  to every classical 
 function on $\cM$ an operator or observable in $End(\cH_0)$.
 The above theorem states that $\cQ$ is surjective for quantum K\"ahler manifolds.
 This means that any given operator $A\in End(\cH_0)$ has a representation
of that form, and in fact many. The kernel of $\cQ$ is typically given by functions above some ``energy cutoff''. 
Furthermore, it follows that the operators of the form \eq{A-diag-rep} form an algebra,
and every operator can be viewed as quantized function on $\cM$.

Even though this is a very nice result,
surjectivity of $\cQ$ is rather surprising in light of the string states 
\eq{string-states}, which are highly non-local. 
Nevertheless, even such string states can be represented in the 
above diagonal form \eq{A-diag-rep},
but $A(x)$ is then rapidly oscillating and in the UV or deep quantum regime.
Therefore this diagonal representation should be used with caution, 
and  a representation in terms of non-local
string states is  more appropriate in the UV regime. 
These can naturally be interpreted as open strings on the embedded 
quantum space or brane $\tilde\cM$.

\paragraph{Completeness relation.}

In particular, the theorem \ref{thm:diag-coh} implies  that at least for 
quantum K\"ahler manifolds, the 
identity operator $\one_{\cH_0}$ can be written in terms of 
coherent states:
\begin{align}
 \one_{\cH_0} = \int_\cM  \one(x) |x\rangle \langle x|,
 \label{one-H0-coherent}
\end{align}
where the integral is defined as in 
\eq{int-def}, and $\one(x)$ is some function on $\cM$.
This gives
\begin{align}
 Tr A &= \int_\cM  \one(x)\langle x|A|x\rangle , \nn\\
 Tr(\cQ(\phi(x))) &= \int_\cM  \one(x)\phi(x) \ .
\end{align}
The natural guess is 
\begin{align}
\one_{\cH} = \int_\cM  |x\rangle \langle x| \ .
% \one(x) = \frac{\dim(\cH_0)}{vol_\omega \cM} \, 1_\cM
\label{complete-simple}
\end{align}
%where $vol_\omega \cM = \int_\cM\Omega$ is the symplectic volume.
This is well-known e.g. for the quantum spaces given by quantized coadjoint orbits of compact 
semi-simple Lie groups, where it follows immediately from Schur's Lemma.
It follows more generally from \eq{X-comm-Q} at least in the semi-classical regime, 
but is not evident if $\one(x)\propto 1_\cM$ for all quantum K\"ahler manifolds.

\section{Remarks and discussion}
\label{sec:remarks}

The results and concepts discussed in this paper call for a number of remarks.

First, we only considered the case where the lowest eigenspace
$E_x$ of $H_x$ is non-degenerate. This excludes many interesting examples
such as fuzzy $S^4_N$ and fuzzy $H^4_n$ as discussed in section \ref{sec:degenerate}. If $E_x$ is an $k$-dimensional (complex) vector space, then 
much of the above analysis would go through, replacing
$\cB$ by an $U(k)$ bundle and $\omega$ by the field strength of its 
natural (Berry) connection.
Sometimes the degeneracy may also be resolved by adding extra matrices 
$X^i$. For example, the abstract quantum space of $S^4_N$ is then recognized as $\C P^{3}$, and similarly in other examples, cf. section \ref{sec:degenerate}. In other words, such degenerate 
quantum spaces can be recognized as projections of non-degenerate ones,
by dropping some $X^a$. 

There are a number of issues which ask for a better understanding.
One of them is the relation between the symplectic volume of $\cM$
and the dimension of the Hilbert space \eq{Tr-Q}. Even though equality 
holds in the standard examples, it is violated for 
the minimal fuzzy torus.
Results from geometric quantization suggest a more complicated 
relation, and it would be desirable to have quantitative 
results for a large class of quantum spaces.
Furthermore, it would be very important to have a more general derivation 
or qualification of the completeness relation \eq{one-approx}.

Another open issue is the compactness of $\cM\subset\C P^{N-1}$ for finite-dimensional $\cH$.
It may be tempting to conjecture that all $\cM$ are compact, but
the fuzzy disk \cite{Lizzi:2003ru} is a candidate for a non-compact 
quantum space, which remains to be elaborated. However, the closure of 
$\cM$ in $\C P^{N-1}$ is clearly compact, and it would be nice to understand 
this in more detail.

Small deformations of the basic quantum K\"ahler manifolds $\cM_0$ of dimension $m<D$ typically lead to an ``oxidation'' $\cM$
corresponding to some tubular neighborhood of 
$\cM_0$. This leads to the idea of fuzzy extra dimensions
\cite{Aschieri:2006uw,Aoki:2014cya}.
On the other hand, it is well-known that 
adding a small perturbation to some quantum manifold $\cM$
 can be viewed 
as a gauge field on $\cM$, which becomes dynamic
in Yang-Mills matrix models. Relating 
this field-theoretic point of view with the above geometric point 
of view provides useful insights,
and one may hope to find further statements on stability and/or 
non-renormalization in this way.
Similar considerations lead to the emergent gravity 
approach based on Yang-Mills matrix models 
\cite{Yang:2006dk,Steinacker:2010rh}.

Finally, the present analysis is restricted to the case of 
irreducible matrix configurations.
If the matrix configuration is reducible, $\cH = \oplus \cH_i$ decomposes into the 
orthogonal sum of irreducible subspaces, and the above considerations apply 
to all $\cH_i$. This could be viewed as a stack of branes.
%However the concept of quasi-coherent states should not be applied to the entire configuration,
%since that would pick out the lowest brane for any given $x$, and may jump between the branes.
In particular, commuting matrix configurations (cf. \cite{Aoki:1998bq})
%are not expected to play an central role in Yang-Mills matrix models: they 
have a large stabilizer  $U(1)^N$ under the adjoint action of $U(N)$,
so that their 
$U(N)$ gauge orbit in Yang-Mills matrix models has smaller dimension 
than that of  irreducible (noncommutative) matrix configurations. But then their  contribution in the ''path`` integral over all matrices is negligible, which defines the quantum theory. 
Therefore irreducible 
matrix configurations as considered here are expected to play the central
role in these models.

\subsection{Dirac operator}

The present framework has a natural extension to spinors
and Dirac-type operators. Namely, for any matrix configuration $X^a,
a=1,...,D$ we can consider \cite{Berenstein:2012ts,deBadyn:2015sca,Karczmarek:2015gda,Schneiderbauer:2016wub}
\begin{align}
  \slashed{D}_x = \Gamma_a(X^a - x^a), \qquad x^a \in \R^D
  \label{Dirac-x}
\end{align}
acting on $\cH \otimes \C^s$. 
Here $\Gamma_a$ are the gamma matrices generating the Clifford algebra of $SO(D)$  on the irreducible representation $\C^s$.
$\slashed{D}_x$ arises as off-diagonal part of the matrix Dirac operator\footnote{ 
A chirality operator for $\slashed{D}$ is typically only recovered in the semi-classical regime.}
$\slashed{D}=\Gamma_a]X^a,.]$ in  Yang-Mills matrix models  such as the IIB or IKKT model, 
for the matrix configuration extended by a point  brane $X^a \oplus x^a$.
It describes a fermionic string stretched between the 
brane and the point $x^a$.
Quite remarkably, numerical investigations \cite{Schneiderbauer:2016wub} strongly suggest that 
 the Dirac operator $\slashed{D}_x$ always has exact zero modes
 \begin{align}
 \slashed{D}_x|x,s\rangle = 0 \ 
\end{align}
 at $\cM$, so that there is no need to introduce 
 the lowest eigenvalue function $\l(x)$. This can be justified rigorously for 2-dimensional branes \cite{Berenstein:2012ts},
 and some heuristic reasons can be given also in more general cases;
 see \cite{Berenstein:2012ts,deBadyn:2015sca,Karczmarek:2015gda} for further work.
 However,  the presence of extra structure 
 due to the spinors obscures the relation with the quasi-coherent states and $\cM$ as 
 introduced here. 
 This is certainly an interesting topic for further research.

\subsection{Effective metric and relation with matrix models}
\label{sec:eff-metric-MM}

The considerations in this paper are motivated by Yang-Mills matrix models, whose 
solutions are precisely matrix configurations as considered here.
Fluctuations in these models are governed by the 
{\em matrix Laplacian}
\begin{align}
 \Box := \d_{ab} [X^a,[X^b,.]]: \quad End(\cH) \to End(\cH) \ .
 \label{Box}
\end{align}
The displacement Hamiltonian arises as off-diagonal part of the matrix Laplacian for a point or probe brane 
added to the matrix configuration
\cite{Schneiderbauer:2016wub}, i.e. for $X^a \oplus x^a$ 
acting on $\cH\oplus\C$. 
It describes a string stretched between the 
brane and the point $x^a$.
This can also be viewed as a special case of 
intersecting branes \cite{Chatzistavrakidis:2011gs}, one brane being the 
point probe.

To understand the effective metric in matrix models,
consider the inner derivations 
\begin{align}
 [X_a,.] \sim i\theta^{a\mu} \del_\mu 
\end{align}
acting on $End(\cH)$ resp. $\cC_{\rm IR}(\cM)$, which
 are (quantizations of)
Hamiltonian vector fields on $\cM$ for almost-local quantum spaces.
By considering the inner product 
$\langle \Phi,\Psi\rangle := Tr([X^a,\Phi^\dagger][X_a,\Psi])$ on $Loc(\cH)$,
one can then show \cite{Steinacker:2010rh} that 
\begin{align}
 \Box \sim e^{\sigma}\Box_G
 \label{Box-G}
\end{align}
where $G$ is the {\bf effective metric} on $\cM$  given by 
\begin{align}
G^{\mu\nu} &= e^{-\sigma}\,\theta^{\mu\mu'} \theta^{\nu\nu'}
 g_{\mu'\nu'}  , \qquad 
  e^{-\sigma} = \frac {|G^{\mu\nu}|^{1/2}}{|\theta^{\mu\nu}|^{1/2}} \ 
\end{align}
for $\dim\cM >2$. This can be viewed as open-string metric, and it
provides the starting point of the emergent geometry and gravity
considerations in \cite{Steinacker:2010rh,Steinacker:2019fcb}.
In the two-dimensional case, the underlying Weyl invariance leads to 
a different interpretation of $\Box$, which is discussed in 
\cite{Arnlind:2012cx}.

In the reducible case, $\cM$ decomposes into a foliation of 
symplectic leaves. Then the effective metric is non-vanishing only 
along this foliation, i.e. it vanishes along the transversal directions.
In the context of Yang-Mills matrix models, this means that fluctuation 
modes on such backgrounds only propagate along the symplectic leaves, 
so that the resulting gauge theory is lower-dimensional.
This happens on any superficially odd-dimensional quantum space, or 
e.g. on $\kappa$ Minkowski space \cite{Lukierski:1993wx} in dimensions larger than 2.

\section{Examples}
\label{sec:examples}

\subsection{The fuzzy sphere}
\label{sec:fuzzy-S2}

The fuzzy sphere 
$S^2_{N}$ \cite{hoppe1982QuaTheMasRelSurTwoBouStaPro,Madore:1991bw} is a 
quantum space defined in terms of three $N \times N$ hermitian matrices
\begin{align}
 X^a =  \frac{1}{\sqrt{C_N}}\, J^a_{(N)}, \qquad \ a=1,2,3
 \label{fuzzy-S2def}
\end{align}
where $J^a_{(N)}$ are the generators of the $N$-dimensional irrep  
of $\msu(2)$ on $\cH = \C^N$, and 
$C_N= \frac 14(N^2-1)$ is the value of the quadratic Casimir.
They satisfy the relations
\begin{equation}
[ X^{{a}}, X^{{b}} ] = \frac{i}{\sqrt{C_N}}\varepsilon^{abc}\, X^{{c}}~ , 
\qquad \sum_{{a}=1}^{3} X^{{a}} X^{{a}} =  \one 
\end{equation}
choosing the normalization  \eq{fuzzy-S2def} such that the radius is one.
The displacement Hamiltonian is
\begin{equation}
H_x=\frac 12\sum_{a=1}^{3}\left(X^{a}-x^{a}\right)^{2}
=\frac 12(\one+|x|^2) -\sum_{a=1}^{3}x^{a}X^{a}
\label{Hx-S2-fuzzy}
\end{equation}
where $|x|^2 = \sum_a x_a^2$.
Using $SO(3)$ invariance, it suffices
to consider the north pole $x=(0,0,x^{3})=:n$
where
\begin{equation}
H_x =\frac 12 (\matunity+|x|^{2}) - |x|\,X^{3}
\label{eq:rotated_lapl}
\end{equation}
assuming $x^{3}>0$ to be specific.
Hence the ground state of $H_x$ is given by the highest weight vector
$|n\rangle := \ket{\frac{N-1}{2},\frac{N-1}{2}}$ of the $\msu(2)$ irrep $\cH$,
and the eigenvalue is easily found to be \cite{Schneiderbauer:2016wub}
\begin{equation}
\l(x)=\frac 12(1+|x|^{2})-|x|\sqrt{\frac{N-1}{N+1}} \ .
\end{equation}
All other quasi-coherent states are obtained by $SO(3)$ acting on $|n\rangle$, hence 
the abstract quantum space $\cM$ is given by the group orbit
\begin{align}
 \cM = SO(3)\cdot |n\rangle =  SO(3)/_{U(1)} \cong S^2\subset \C P^{N-1} \ .
\end{align}
Note that the quasi-coherent states are constant along the radial lines 
in agreement with \eq{V-anihil},
\begin{align}
 |x\rangle = |\a x\rangle\qquad \mbox{for} \quad \a > 0 \ .
\end{align}
The equivalence classes $\cN$ consist of the  radial lines emanating
from the origin, and the would-be symplectic form $\omega_{ab}$ and the quantum metric $g_{ab}$  vanish if
any one component is radial. 
The minima of $\l(x)$ on $\cN_x$ describe a sphere with radius
$|x_{0}|=\sqrt{\frac{N-1}{N+1}}=1+\mathcal{O}(\frac{1}{N})$.
This coincides precisely with the  embedded quantum space \eq{embedded-M-def} 
\begin{equation}
\tilde\cM = \{\langle x|X^a|x\rangle\}
 = \big\{x\in\realv^{3}:\,|x|=\sqrt{\frac{N-1}{N+1}} \big\} \cong S^2 \ 
\end{equation}
defined by the expectation value ${\bf x^a}$ \eq{M-embedding-symbol}, in accordance with \eq{X-X0-proj}.
At the singular set $\cK=\{0\}$ the Hamiltonian is
$H_0 = C^2 \one$, so that all energy levels become degenerate and cross. Following 
$|x\rangle$ along the radial direction through the origin, it 
turns into the highest energy level. 
It is easy to see that the would-be symplectic form $\omega$ is the unique $SO(3)$-invariant 2-form on $\cM$
which satisfies the quantization condition \eq{quant-cond-S2} with $n=N$.
Moreover, the abstract quantum space $\cM\cong S^2\subset\C P^{N-1}$
is a quantum K\"ahler manifold, since the complex tangent space \eq{tangent-space-C} is 
one-dimensional, spanned   by 
\begin{align}
 T_{n,\C}\cM = \big\langle J^-|n\rangle \big\rangle_\C \
\end{align}
(at $|n\rangle\in\cM$).
This holds because  $|n\rangle$ is the highest weight state, so that 
\begin{align}
 J^+|n\rangle = 0 \ ;
\end{align}
therefore the two tangent vectors $\cX^1|n\rangle, \cX^2|n\rangle \in T_{n}\cM$ \eq{tangent-space-R}
are related by $i$, while  $\cX^3|n\rangle$ vanishes at $n$.
%%H V2
Indeed, it is well-known that the coherent states on $S^2_N$ form a Riemann sphere,
and the (quasi-) coherent states coincide with the coherent states
introduced in \cite{Perelomov:1986tf}.

All this holds for any $N\geq 2$. The coherence length is of order 
\begin{align}
 L_{\rm coh} \approx L_{NC} \sim \frac{1}{\sqrt{N}} \ 
\end{align}
in the given normalization. Hence
for sufficiently large $N$, the almost-local operators 
comprise all polynomials in $X^a$ up to order $O(\sqrt{N})$
(depending on some specific bound), so that $S^2_N$ 
is  an almost-local quantum space. 
In contrast for the {\bf minimal fuzzy sphere} $S^2_2$
with $N=2$, the generators reduce to the Pauli matrices $X^a = \sigma^a$,
and the (quasi)coherent states form the well-known Bloch sphere 
$\cM = S^2\cong\C P^1$. 
This is still a quantum K\"ahler manifold even though the semi-classical regime 
is trivial and contains only the constant functions $Loc(\cH) = \C\one$, 
since the coherence length is of the same order 
as the entire space $\cM$.

\subsection{Quantized coadjoint orbits for compact semi-simple Lie groups}
\label{sec:QCO}

The above construction generalizes naturally to quantized coadjoint orbits for
any compact semi-simple Lie group $G$ with Lie algebra $\mg$. 
For any irreducible representation $\cH_\L$
with highest weight $\L = (n_1,...,n_k)$ labeled by Dynkin indices $n_j$,
the matrix configuration 
\begin{align}
 X^a = c\, T^a, \qquad a=1,...,D
 \label{qco-config}
\end{align}
defines a  quantum K\"ahler manifold $\cM \cong G/K$. 
Here $T^a$ are orthogonal generators of $\mg \cong \R^D$ acting on $\cH_\L$,
$K$ is the stability group of the highest weight $\L$, and
$c$ is some normalization constant. 
Then  the displacement Hamiltonian is 
\begin{align}
 H_x = C^2(\mg) + \frac 12 x_a x^a   - x_a T^a\ 
\end{align}
where $C^2(\mg) \propto \one$ is the quadratic Casimir.
Using $G$-invariance, we can  assume that $x$ is in  
(the dual of) the Cartan subalgebra and has maximal weight. Then  $|x\rangle = |\L\rangle$
is the highest weight state, so that
the quasi-coherent states are the group orbit
$\cM = G\cdot |\L\rangle \cong G/K$
of the highest weight state with stabilizer $K$. This is  a quantum K\"ahler manifold 
due to the highest weight property,
%%H V2
and the quantum metric $g_{ab}$ \eq{g-def} and the 
symplectic form $\omega$ \eq{omega-def} 
are the canonical group-invariant structures on the K\"ahler manifold $\cM$.
For large Dynkin indices $n_j\geq n\gg 1$, 
the almost-local operators 
comprise  all polynomials in $X^a$ up to some order $O(\sqrt{n})$, 
so that $\cM$ is  an almost-local quantum space.
This is essentially the well-known story of quantized 
coadjoint orbits, and the (quasi-) coherent states coincide with the coherent states
introduced in \cite{Perelomov:1986tf}, cf. \cite{Grosse:1993uq}. Perhaps less known is the fact that 
if some of the $n_j$ are small, 
$\cM$ can be viewed as ``oxidation'' of some lower-dimensional brane, 
more precisely as a bundle over $\cM_0$ whose fiber is 
very ``fuzzy''. For an  application of such a structure 
see e.g. section 4.2 in \cite{Sperling:2018hys}.

This construction generalizes further to  highest weight (discrete series) 
unitary irreducible representation of non-compact semi-simple Lie groups.
A particularly interesting example is given by the ``short'' series of unitary irreps of $SO(4,2)$
known as singletons, which lead to the fuzzy 4-hyperboloids $H^4_n$ discussed below, and 
 to quantum spaces which can be viewed as cosmological space-time \cite{Sperling:2019xar}.

\paragraph{(Minimal) fuzzy $\C P^{N-1}_N$.}

As an  example we consider minimal fuzzy $\C P^{N-1}_N$,
which is obtained using the above general construction for 
$G=SU(N)$ and its fundamental representation $\cH = (1,0,...,0)$, so that
$G/K \cong \C P^{N-1}$.
This is the quantum K\"ahler manifold obtained from the matrix configuration
\begin{align}
 X^a = \l^a \quad \in End(\cH), \qquad \cH = \C^N
\end{align}
for $a=1,...,N^2-1$,
where $\l^a$ are a (Gell-Mann) ON  basis  of $\msu(N)$
in the fundamental representation. 
Then
$End(\cH)\cong (0,...,0) \oplus (1,0,...,0,1)$ can be viewed as a minimal 
quantization of functions on $\C P^{N-1}$.
The quantization map
\begin{align}
 \cQ(\phi) = \int_{\C P^{N-1}}|x\rangle\langle x| \phi(x)
\end{align}
is then the partial 
inverse of the symbol map, apart from the constant function:
\begin{align}
 \cQ(\langle x|\Phi|x\rangle) = c \Phi\qquad \mbox{if} \ \ Tr(\Phi) = 0 \ 
\end{align}
for some $c>0$.
Near $|\L\rangle$, 
the quasi-coherent states $|x\rangle$ can be organized as holomorphic sections 
\begin{align}
 \|z\rangle = \exp(z^k T^+_k)|\L\rangle \ ,
\end{align}
where 
the $T^+_k, \ k=1,...,N-1$ are the rising operators of a Chevalley basis of 
$\msu(N)$. Hence fuzzy $\C P^{N-1}_N$ is a quantum K\"ahler manifold
which coincides with $\C P^{N-1}$, with
K\"ahler form
\begin{align}
 \omega_{\bar k l} = \frac{\del}{\del \bar z^k}\langle z\|\frac{\del}{\del z^l}\|z\rangle \ .
\end{align}

\paragraph{Squashed $\C P^2_N$.}

Further quantum spaces can be obtained by  projections of 
quantized coadjoint orbits. For example, starting from 
fuzzy $\C P^2_N$ with $\cH = (N,0)$, consider the following 
matrix configuration
\begin{align}
 X^a = T^a, \qquad a=1,2,4,5,6,7
\end{align}
dropping the Cartan generators $T_3$ and $T_8$ from the (Gell-Mann) basis
of $\msu(3)$. Then the displacement Hamiltonian can be written 
as
\begin{align}
 H_x =  \bar H_x - \frac 12(X_3-x_3)^2 -  \frac 12(X_8-x_8)^2 
 %= (\sum_i (X_i^+ - \bar z_i)(X_i^- - z_i)
 \label{disp-H-squashed}
\end{align}
where $\bar H_x$ is the displacement Hamiltonian for $\C P_N^2$.
Although the quasi-coherent states $|x\rangle$ are not known in this case,
they are close to those  of $\C P^2_N$ in the large $N$ limit,
cf. \cite{Schneiderbauer:2016wub}.
Indeed then the last two terms in \eq{disp-H-squashed} are small, and 
$0 < \l(x) \leq \bar \l(x)$ gives an upper bound for $\l$. 
This implies that the displacement is  small, and
\begin{align}
 \cM \approx \C P^2 \subset \C P^{N(N+3)/2} \ .
\end{align}
Again, the concept of
the abstract quantum space is  superior to the notion of an embedded brane,  which is a complicated self-intersecting variety in $\R^6$
related to the Roman surface \cite{Steinacker:2014lma}.

\subsection{Degenerate cases}
\label{sec:degenerate}

\paragraph{The fuzzy 4-sphere  $S^4_N$.}

Now consider again the quantized coadjoint orbit of $SU(4)\cong SO(6)$
acting on the highest weight irrep $\cH_\L$ with $\L=(N,0,0)$. We have seen just 
that the matrix configuration using 
all $\mso(6)$ generators $\cM^{ab} = - \cM^{ba}$ as in \eq{qco-config}
would give fuzzy 
$\C P^3_N$, with coherent states acting on the highest weight 
state $|\L\rangle$. Now instead of using all  $\cM^{ab}$, 
consider the matrix configuration defined by the following 5 hermitian matrices
\begin{align}
 X^a = \cM^{a6} \qquad \in End(\cH_\L), \qquad a = 1,...,5 \ .
\end{align}
Using $SO(5)$ invariance, it suffices to consider the 
%Since $\cH_\L$ remains irreducible under $SO(5)\subset SO(6)$, 
displacement Hamiltonian at $x=(0,0,0,0,x_5)$,  
\begin{align}
 H_x = \frac 12\sum_{i=1}^4 X_i^2 + \frac 12 (X_5 - x_5)^2
  = \frac 12 (R^2 + x_5^2) \one - x_5 X^5 
\end{align}
since $ \sum_a X_a^2 = R^2\one$ for $R^2 = \frac 14 N(N+4)$, cf. \cite{Grosse:1996mz,Medina:2012cs}.
Now $|\Lambda\rangle$ is by construction an eigenstate of $X^5$
which commutes with $SO(4)$, with maximal eigenvalue.
Therefore the lowest eigenspace $E_x$ of $H_x$ is spanned by the 
orbit $SO(4)\cdot |\L\rangle \cong S^2$, which spans 
a $N+1$-dimensional complex vector space. This provides an example of 
a degenerate quantum space. The abstract quantum space
$\cM$ is obtained by acting with $SO(5)$ on this $S^2$, which is easily 
seen to recover 
\begin{align}
 \cM \cong \C P^3 \ \subset \C P^{\dim\cH -1}
\end{align}
which is an equivariant $S^2$ bundle over $S^4$.
The $E_x$ naturally form a $SU(N+1)$ bundle $\cB$ over $S^4$,
and $\omega$ is replaced by an $SU(N+1)$ connection.
Again the concept of an abstract quantum space greatly helps 
to understand the structure, as it resolves the degeneracy of the 
quasi-coherent states. Moreover $\cM$ is clearly a K\"ahler manifold,
and theorem \ref{thm:diag-coh} holds.

\paragraph{The fuzzy 4-hyperboloid $H^4_n$.}

Using an analogous construction for $SO(4,2)$ and its singleton 
irreps $\cH_n$ labeled by $n\in\N$, one obtains 
fuzzy $H^4_n$ \cite{Sperling:2018xrm,Hasebe:2012mz}. The corresponding matrix
configuration is given by the following 5 hermitian operators
\begin{align}
 X^a = \cM^{a5} \qquad \in End(\cH_\L), \qquad a = 0,...,4 \ .
\end{align}
However, it is more appropriate here to define the displacement Hamiltonian
using $\eta_{ab}$, so that $SO(4,1)$ is preserved. Then 
we can assume that $x = (x_0,0,0,0,0)$, so that
\begin{align}
 H_x = \frac 12\sum_{i=1}^4 X_i^2 - \frac 12 (X_0 - x_0)^2
  = \frac 12 (R^2 - x_0^2) \one + x_0 X^0 \ .
\end{align}
Then the resulting quasi-coherent states form an abstract quantum space 
$\cM \cong\C P^{1,2}$, which is an $S^2$ bundle over $H^4$.
It is a K\"ahler manifold, and theorem \ref{thm:diag-coh} 
 still holds in a weaker sense \cite{Sperling:2018xrm}.
%It becomes an almost-local quantum space for large $n$.
This in turn is the basis of the cosmological space-time solution 
$\cM^{3,1}_n$ with an effective metric of FLRW type, as 
discussed in \cite{Sperling:2019xar,Steinacker:2019awe}.

\paragraph{Minimal fuzzy $H^4_0$.}

A particularly interesting example 
is obtained from $H^4_n$ for $n=0$, which is not a 
quantized coadjoint orbit and not even symplectic.
In that case $E_x$ is one-dimensional, and one can check that 
$\langle x|\del_a|x\rangle = 0 = i A_a$ and 
$\langle x|[X^a,X^b]|x\rangle = 0$. Therefore 
the would-be symplectic form $\omega$ vanishes. The abstract quantum space
is then
\begin{align}
 \cM = H^4
\end{align}
but it carries a trivial line bundle $\tilde\cB$. It  still satisfies the  quantum K\"ahler\footnote{Note that $\dim\cH = \infty$ here, 
so that we cannot conclude that $\cM$ is K\"ahler in the usual sense.}
condition \eq{Kahler-cond} and theorem \ref{thm:diag-coh}  
should hold (using the $SO(4,1)$-invariant integral) in a weaker sense.
However this is not an almost-local quantum space, 
and there is no semi-classical regime.

\subsection{The minimal fuzzy torus}
\label{sec:min-fuzzy-T2}

The minimal fuzzy torus $T^2_2$ turns out 
to be a quantum manifold which is not K\"ahler, and not even symplectic.
It is defined in terms of 
\begin{align}
 U = \begin{pmatrix}
      0 & 1 \\
      1 & 0
     \end{pmatrix} = X_1 + i X_2, \qquad
V = \begin{pmatrix}
     1 & 0 \\
     0 & -1
    \end{pmatrix} = X_3 + i X_4
\end{align}
which defines 4 hermitian matrices $X_i = X_i^\dagger \in End(\C^2)$.
Noting that $[U,U^\dagger] = 0 = [V,V^\dagger]$ and
\begin{align}
(U-z)(U-z)^\dagger  
%&=  \begin{pmatrix}
%       -z & 1 \\
%       1 & -z
%      \end{pmatrix} 
%      \begin{pmatrix}
%       -z^* & 1 \\
%       1 & -z^*
%      \end{pmatrix}
&= \begin{pmatrix}
        1 + |z|^2 & -z - z^* \\
        -z - z^* & 1 + |z|^2
       \end{pmatrix}   \nn\\
(V-w)(V-w)^\dagger 
%  &=\begin{pmatrix}
%      1-w & 0 \\
%      0 & -1-w
%     \end{pmatrix}\begin{pmatrix}
%      1-w^* & 0 \\
%      0 & -1-w^*
%     \end{pmatrix}
  &=\begin{pmatrix}
       |1-w|^2 & 0 \\
       0 & |1+w|^2
      \end{pmatrix}
\end{align}
where $z = x_1 + i x_2$ and $w = x_3 + i x_4$,
the displacement Hamiltonian is
\begin{align}
 H_y &= \sum(X_i - x_i)^2 
  =  \begin{pmatrix}
        1 + |z|^2 + |1-w|^2 & -z - z^* \\
       -z - z^* & 1 + |z|^2 + |1+w|^2
      \end{pmatrix} \ .
\end{align}
% The spectrum of this Hamiltonian is found to be
% \begin{align}
%  \l_{1,2} &= 2+|z|^2+|w|^2 \pm \sqrt{|z+z^*|^2 + |w+w^*|^2}, 
% \end{align}
%so that 
The lowest eigenvalue is 
\begin{align}
 \l = 2+|z|^2+|w|^2 - \sqrt{|z+z^*|^2 + |w+w^*|^2}
\end{align}
and the corresponding quasi-coherent states are
\begin{align}
 |x\rangle \propto \begin{pmatrix}
              \sqrt{|z+z^*|^2 + |w+w^*|^2} +w^*+w \\
              z^*+z
             \end{pmatrix} \quad \in \R_+ \times \R \subset \C^2 \ .
            \label{state-T2-2}
\end{align}
These clearly depend only on the real parts of $z, w$, and 
the normalized states %$\langle x|x\rangle = 1$ 
describe a half circle  in the upper half plane.
However
the two endpoints of this half-circle corresponding to $(z=1,w=-\infty)$ and $(z=-1,w=-\infty)$
describe the same state $|x\rangle = \begin{pmatrix}
                                         0 \\ \pm 1
                                        \end{pmatrix}$, and should hence 
                                        be identified. 
Thus $\cM = S^1$, which
is clearly not a K\"ahler manifold any not even symplectic.

Now consider the equivalence classes $\sim$ \eq{cM-equivalence-def} 
on $\R^4  \cong \C^2$.
All points $(z,w)\sim (z',w')\in\C^2$ with the same real parts 
are identified, 
and also  all $(z,w) \sim r(z,w) \in \R^2$ for $r>0$.
Among these, $\l$ assumes the minimum $\l=1$ for $(z,w) = (x,y) \in S^1 \subset \C^2$,
so that again\footnote{
It may seem that the state 
corresponding to the point $(z=0,w=-1)$ vanishes, but this is just 
an artefact of the improper normalization. It is easy to see that
in that case $H_y$ has indeed an eigenstate $(0,1)$ for $\l=1$.}
 $\cM \cong \C^2/_\sim \cong S^1$.

Therefore the minimal fuzzy torus $T^2_2$ should really be considered 
as a fuzzy circle.
This shows the existence of  
''exotic`` quantum spaces which are not quantized symplectic spaces,
but do not have a semi-classical regime.
There are also higher-dimensional such spaces as shown next, and 
  the above example of minimal $H^4_0$.

\paragraph{Non-K\"ahler quantum space from $T_2^2 \times T_2^2$.}

Now consider the Cartesian product of $T^2_2 \times T^2_2$,
realized through $8$ hermitian matrices $X^a_{(1)}, X^a_{(2)}$ acting on $\C^4 = \C^2 \otimes \C^2$.
All eigenstates of $H_x = H_x^{(1)} + H_x^{(2)}$ are  given by the product states
of the two eigenstates \eq{state-T2-2} of $T^2_2$, so that the ground states or quasi-coherent states are given by
\begin{align}
|x_{(1)},x_{(2)}\rangle = |x_{(1)}\rangle \otimes|x_{(2)}\rangle 
\end{align}
over $\R^8$. They are again degenerate, and inequivalent states are parametrized by 
$(x_{(1)},x_{(2)}) \in S^1 \times S^1$.
Hence the abstract quantum space is a torus $\cM \cong S^1 \times S^1$. 
The quantum tangent space is spanned by  two vectors
\begin{align}
 T_\xi\cM = \Big\langle (\del_1 |y_{(1)}\rangle) \otimes|y_{(2)}\rangle,
 |y_{(1)}\rangle \otimes (\del_2|y_{(2)}\rangle) \Big\rangle \ \cong \ \R^2
\end{align}
which are linearly independent from the two complexified vectors
$i\del_1 |y_{(1)}\rangle \otimes|y_{(2)}\rangle$ and $ i|y_{(1)}\rangle \del_2\otimes|y_{(2)}\rangle$.
Therefore $T_{\xi,\C}\cM \cong \C^2 \cong \R^4$, and
$\cM$ is not a quantum K\"ahler manifold.

\subsection{The Moyal-Weyl quantum plane}
\label{sec:Moyal-Weyl}

The Moyal-Weyl quantum plane is obtained for $X_1 = X$ and $X_2 = Y$
with $[X,Y] = i\one$. Then $\dim\cH=\infty$, 
but all  considerations can be carried over easily.
The displacement Hamiltonian 
\begin{align}
 2 H_x = (X-x)^2 + (Y-y)^2
\end{align}
is nothing but the shifted harmonic oscillator, with ground state
\begin{align}
  H_z |z\rangle = \frac 12 |z\rangle \
\end{align}
given by the standard coherent states 
\begin{align}
 |z\rangle &= U(z)|0\rangle, \qquad z = \frac{1}{\sqrt{2}}(x+iy) 
\end{align}
using the identification of $\R^2 \cong\C$. The translation operator is given 
as usual by 
\begin{align} 
 U(z) &= \exp(i(yX - xY)) = \exp(z a^\dagger - \bar z a)  ,  \nn\\
 a &= \frac{1}{\sqrt{2}}(X+iY), \qquad a^\dagger = \frac{1}{\sqrt{2}}(X-iY) \ .
\end{align}
$|0\rangle$ is the ground state of the harmonic oscillator $a|0\rangle = 0$, 
and more generally
\begin{align}
 (a-z)|z\rangle &= 0  
\end{align}
implies
\begin{align}
 \langle z |(X+iY)|z\rangle &= x+iy \ .
\end{align}
The derivatives \eq{del-nabla-X-general} are found to be
\begin{align}
 (\del_x -iA_1) |z\rangle &= -i(Y - y) |z\rangle = i\cX_1|z\rangle  \nn\\
 (\del_y-iA_2) |z\rangle &= i(X -  x) |z\rangle = i\cX_2|z\rangle
 \label{transl-Moyal}
\end{align}
where the second expressions arise from \eq{del-nabla-X-general-2}, 
which are given explicitly by 
\begin{align}
 \cX_1 &= -i[\big(H_z-\frac 12\big)^{'-1},X]   \nn\\
 \cX_2 &= -i[\big(H_z-\frac 12\big)^{'-1},Y] \ .
\end{align}
The $U(1)$ connection is found to be 
\begin{align}
 iA_1 = \langle z| \del_x |z\rangle &= -i \langle z| (Y-\frac 12 y) |z\rangle = -\frac i2 y  \nn\\
 iA_2 = \langle z| \del_y |z\rangle &= i \langle z| (X - \frac 12 x) |z\rangle = \frac i2 x
\end{align}
with field strength
\begin{align}
 F_{12} = \del_1 A_2 - \del_2 A_1 = 1 \ .
\end{align}
% To check this, we need 
% \begin{align}
%  \cX_1|x,y\rangle &= \big(2(H_x-1)^{'-1}(X-x) -iy\big)|x,y\rangle \nn\\
%   &= \big((X-x) -iy\big)|x,y\rangle \nn\\
%   &= \big(-i(Y-y) -iy\big)|x,y\rangle \nn\\
%   &= -iY|x,y\rangle
% \end{align}
% noting that $H_x (X-x)|x,y\rangle = 3(X-x)|x,y\rangle$ etc.,
% and using the annihilation property \eq{}, in agreement with \eq{transl-Moyal}.
Therefore \eq{parallel-transport-states} becomes
\begin{align}
 |z\rangle = P\exp\Big(i\int_0^{z} (\cX_1-y) dx + (\cX_2+x) dy\Big)|0\rangle \ .
\end{align}
% We verify
% \begin{align}
%  [\cY_x,\cY_y]|x,y\rangle &= 4\big((H_x-1)^{'-1}(X-x)(H_x-1)^{-1}(Y-y)
%   - (H_x-1)^{'-1}(Y-y)(H_x-1)^{-1}(X-x) \big) |x,y\rangle \nn\\
%    &= 2(H_x-1)^{'-1}\big((X-x)(Y-y) - (Y-y)(X-x) \big) |x,y\rangle \nn\\
%    &= 2 i (H_x-1)^{'-1} |x,y\rangle = 0
% \end{align}
$\cM\cong\C$ satisfies the quantum K\"ahler condition due to
the constraint $(X+iY)|0\rangle = 0$, which states that 
$iY|0\rangle = - X|0\rangle$, so that the complex tangent space 
$T_{0,\C}\cM = T_0\cM$ coincides with the real one.
 The holomorphic coherent states 
are given by 
\begin{align}
 \|z\rangle = e^{z a^\dagger}|0\rangle 
  = e^{z a^\dagger}e^{-\bar z a}|0\rangle 
 % = e^{\frac 12|z|^2} e^{z a^\dagger -\bar z a}|0\rangle 
   = e^{\frac 12|z|^2} |z\rangle  \ .
\end{align}
They cannot be  normalized, since 
the map $z\mapsto\langle w\|z\rangle$ must be holomorphic and hence 
unbounded. Thus $\|z\rangle$ should be viewed as holomorphic section of 
the line bundle $\tilde\cB$.

\subsection{Commutative quantum spaces}
\label{sec:comm}

In the infinite-dimensional case, one can also consider
matrix configurations associated to commutative manifolds.
The simplest example is the circle $S^1$, which arises from the single 
operator
\begin{align}
 X = -i\del_\varphi
\end{align}
acting on $\cC^\infty(S^1)  \subset  L^2(S^1) = \cH$.  
The displacement Hamiltonian is
\begin{align}
 H_x = \frac 12 (-i\del_\varphi - x)^2 , \qquad x \in \R \ .
\end{align}
The quasi-coherent states for $x=n\in\Z$  are clearly
\begin{align}
 |n\rangle = e^{i n \varphi}, \qquad  H_n |n\rangle = 0 , \qquad n\in\Z
 \label{quasicoh-del}
\end{align}
so that $\l(\Z)=0$.  For any $x\not\in\Z$, 
all eigenstates of $H_x|\psi\rangle = E|\psi\rangle$
are given by the above states $|n\rangle$, with eigenvalue
\begin{align}
 H_x |n\rangle = (-i\del_\varphi - x)^2 e^{i n \varphi}
  = (n - x)^2 e^{i n \varphi} \ .
\end{align}
Therefore 
\begin{align}
 |x\rangle =  |n\rangle , \qquad  |n-x| < \frac 12 , \ \ n\in\Z
\end{align}
while for $x\in\Z+\frac 12$ the space $E_x$ is two-dimensional, containing
both states $|x\pm\frac 12\rangle$.
Thus the abstract quantum space is  the discrete lattice 
\begin{align}
 \cM = \Z \subset \R
\end{align}
and the quantum tangent space vanishes. 
This can be generalized to the higher-dimensional 
commutative torus $T^n$ with commutative and reducible
matrix configuration $X_\mu = -i\del_\mu$,
which  also  leads to a discrete quantum space
without further structure.
Thus classical manifolds are not well captured in the 
present framework. This can of course be treated by adding extra structure 
as in \cite{connes1995noncommutative}, but such a description is not well suited for 
Yang-Mills matrix models.

\paragraph{Acknowledgments}

I would like to thank Bernhard Lamel and Thorsten Schimannek for useful discussions and pointing me to the appropriate literature.
Related collaborations and discussions with 
Timon Gutleb, Joanna Karczmarek, Jun Nishimura, 
Lukas Schneiderbauer and Jurai Tekel are gratefully acknowledged.
Finally, John Madore's role in shaping the underlying ideas is greatly appreciated.
This work was supported by the Austrian Science Fund (FWF) grant P32086.

\section{Conclusion}

%%H V2 added
A general framework for quantum geometry was developed, 
based on general matrix configurations given in terms of $D$ hermitian 
matrices $X^a$. 
%This framework is 
% particularly relevant in the context of Yang-Mills matrix models
%such as the IIB matrix model, and more generally for ''almost-commuting``
%matrix configurations.
%%
We have seen that a remarkably rich array of structures can be extracted from 
such a  matrix configuration, which provide a semi-classical picture and geometric insights. Quasi-coherent states are an optimal set of states where 
the matrices  are simultaneously ''almost-diagonal``. 
They form an abstract quantum space 
$\cM\subset\C P^{N}$, which 
allows to use geometric tools and even complex analysis. 
A class of almost-local operators $Loc(\cH)$ is characterized,
which can be understood as quantized functions on $\cM$ in some IR regime.
Moreover, a natural sub-class of matrix configurations is 
identified as quantum K\"ahler manifolds.

Although the present analysis is restricted to the case of 
finite-dimensional matrices, the concepts  
generalize to the case of selfadjoint operators on 
separable Hilbert spaces. 
This is illustrated for the Moyal-Weyl quantum plane 
 and for the fuzzy hyperboloid. In these cases,
the framework exhibits  the 
finite number of degrees of freedom per unit volume, as well as the stringy nature in the deep quantum regime. It should also be useful 
to better understand other quantum spaces such as $\kappa$ Minkowski space \cite{Lukierski:1993wx}, and to resolve a hidden internal structure 
in other spaces such as \cite{Fiore:2019rgy}
%One can then extract the abstract quantum space $\cM$ and 
%the associated geometric structure without imposing any extra structure.
and in compact quantum spaces with infinite-dimensional $\cH$.

This framework for quantum geometry is 
particularly suited for Yang-Mills-type matrix models. 
Their description in terms of quantized symplectic spaces 
is now understood to be generic, rather than just an ad-hoc choice.
This vindicates describing the low-energy regime of such matrix models
via noncommutative field theory on the embedded 
quantum space or brane $\tilde\cM$, leading to dynamical emergent geometry 
and possibly gravity, cf. \cite{Steinacker:2010rh,Steinacker:2019awe}. However,
it is important to keep in mind that semi-classical picture 
breaks down in the UV or deep quantum regime, 
where non-local string states become dominant.
These are naturally  interpreted as open strings on the brane $\tilde\cM$.
% Moreover, the interpretation of the matrices is not evident: they might 
% be interpreted as position space coordinates, or as momentum coordinates as in 
% \cite{Sperling:2019xar}.

In particular, the new insights on the structure of $\cM$
should be very useful to interpret the results of numerical simulations 
of Yang-Mills matrix models 
\cite{Nishimura:2019qal,Aoki:2019tby,Kim:2011cr,Anagnostopoulos:2020xai}. By definition, 
the quasi-coherent states 
provide an optimal basis where the matrices are ''almost-diagonal``,
which should improve upon  simpler approaches based on 
block-matrices. They can be obtained numerically 
along the lines proposed in \cite{Schneiderbauer:2016wub,lukas_schneiderbauer_2016_45045}, which can now be 
refined, notably using the abstract 
point of view as $\cM\subset\C P^{N-1}$. It should then be  
easier to disentangle the underlying geometry
from the random noise.

The framework should also be useful 
for analytical 
computations in the context of  noncommutative field theory.
Given the natural role of quantum K\"ahler manifolds in this setting,
one may hope that quantum K\"ahler manifolds 
play a special and preferred role not only from an analytical point of view,
but also as preferred solutions or configurations in  a matrix ``path integral''. For example, loop integrals analogous to \eq{complete-simple} can be formulated in terms of 
the completeness relation for string states  
\cite{Steinacker:2016nsc}.
In particular, one may hope that some sort of non-renormalization statement 
can be made on such spaces.

Finally, it would be desirable to improve some the technical
results in this paper, notably related to the completeness relation and the
regularity of $\cM$.
In particular, one would like to know to which extent the results on 
quantum K\"ahler manifolds can be generalized to generic quantum manifolds
with symplectic structure and a metric.
It would also be interesting to develop an analogous approach 
based on the 
matrix Dirac operator as sketched in section \ref{sec:remarks}, 
and to relate it to the present approach.
All these are interesting directions for future work.

\bibliography{papers}

\begin{thebibliography}{10}

\bibitem{Madore:1991bw}
J.~Madore,
\newblock {\em {The Fuzzy sphere}},
\newblock Class. Quant. Grav. {\bf 9}, 69 (1992).

\bibitem{hoppe1982QuaTheMasRelSurTwoBouStaPro}
J.~R. {Hoppe},
\newblock {\em {Quantum Theory of a Massless Relativistic Surface and a
  Two-Dimensional Bound State Problem.}},
\newblock PhD thesis, Massachusetts Institute of Technology, 1982.

\bibitem{Arnlind:2010kw}
J.~Arnlind, J.~Hoppe, and G.~Huisken,
\newblock {\em {Discrete curvature and the Gauss-Bonnet theorem}},
\newblock (2010), 1001.2223.

\bibitem{Steinacker:2014lma}
H.~C. Steinacker and J.~Zahn,
\newblock {\em {Self-intersecting fuzzy extra dimensions from squashed
  coadjoint orbits in $ \mathcal{N}=4 $ SYM and matrix models}},
\newblock JHEP {\bf 02}, 027 (2015), 1409.1440.

\bibitem{Ishibashi:1996xs}
N.~Ishibashi, H.~Kawai, Y.~Kitazawa, and A.~Tsuchiya,
\newblock {\em {A Large N reduced model as superstring}},
\newblock Nucl. Phys. B {\bf 498}, 467 (1997), hep-th/9612115.

\bibitem{Glaser:2019lcd}
L.~Glaser and A.~B. Stern,
\newblock {\em {Reconstructing manifolds from truncated spectral triples}},
\newblock (2019), 1912.09227.

\bibitem{connes1995noncommutative}
A.~Connes,
\newblock {\em Noncommutative Geometry} (Elsevier Science, 1995).

\bibitem{Nishimura:2019qal}
J.~Nishimura and A.~Tsuchiya,
\newblock {\em {Complex Langevin analysis of the space-time structure in the
  Lorentzian type IIB matrix model}},
\newblock JHEP {\bf 06}, 077 (2019), 1904.05919.

\bibitem{Aoki:2019tby}
T.~Aoki, M.~Hirasawa, Y.~Ito, J.~Nishimura, and A.~Tsuchiya,
\newblock {\em {On the structure of the emergent 3d expanding space in the
  Lorentzian type IIB matrix model}},
\newblock PTEP {\bf 2019}, 093B03 (2019), 1904.05914.

\bibitem{Kim:2011cr}
S.-W. Kim, J.~Nishimura, and A.~Tsuchiya,
\newblock {\em {Expanding (3+1)-dimensional universe from a Lorentzian matrix
  model for superstring theory in (9+1)-dimensions}},
\newblock Phys. Rev. Lett. {\bf 108}, 011601 (2012), 1108.1540.

\bibitem{Anagnostopoulos:2020xai}
K.~N. Anagnostopoulos {\em et~al.},
\newblock {\em {Complex Langevin analysis of the spontaneous breaking of 10D
  rotational symmetry in the Euclidean IKKT matrix model}},
\newblock JHEP {\bf 06}, 069 (2020), 2002.07410.

\bibitem{Ishiki:2015saa}
G.~Ishiki,
\newblock {\em {Matrix Geometry and Coherent States}},
\newblock Phys. Rev. {\bf D92}, 046009 (2015), 1503.01230.

\bibitem{Schneiderbauer:2016wub}
L.~Schneiderbauer and H.~C. Steinacker,
\newblock {\em {Measuring finite Quantum Geometries via Quasi-Coherent
  States}},
\newblock J. Phys. {\bf A49}, 285301 (2016), 1601.08007.

\bibitem{Berenstein:2012ts}
D.~Berenstein and E.~Dzienkowski,
\newblock {\em {Matrix embeddings on flat $R^3$ and the geometry of
  membranes}},
\newblock Phys. Rev. {\bf D86}, 086001 (2012), 1204.2788.

\bibitem{Ishiki:2016yjp}
G.~Ishiki, T.~Matsumoto, and H.~Muraki,
\newblock {\em {Kahler structure in the commutative limit of matrix geometry}},
\newblock JHEP {\bf 08}, 042 (2016), 1603.09146.

\bibitem{deBadyn:2015sca}
M.~H. de~Badyn, J.~L. Karczmarek, P.~Sabella-Garnier, and K.~H.-C. Yeh,
\newblock {\em {Emergent geometry of membranes}},
\newblock JHEP {\bf 11}, 089 (2015), 1506.02035.

\bibitem{Karczmarek:2015gda}
J.~L. Karczmarek and K.~H.-C. Yeh,
\newblock {\em {Noncommutative spaces and matrix embeddings on flat
  $\mathbb{R}^{2n+1}$}},
\newblock JHEP {\bf 11}, 146 (2015), 1506.07188.

\bibitem{lukas_schneiderbauer_2016_45045}
L.~Schneiderbauer,
\newblock Bprobe: a wolfram mathematica package, 2016,
  \href{http://dx.doi.org/10.5281/zenodo.45045}{doi:10.5281/zenodo.45045}.

\bibitem{lukas_schneiderbauer_2019_2616687}
L.~Schneiderbauer and T.~S. Gutleb,
\newblock Tsgut/bprobem: v1.0.5, 2019,
  \href{http://dx.doi.org/10.5281/zenodo.2616687}{doi:10.5281/zenodo.2616687}.

\bibitem{rellich1969perturbation}
F.~Rellich and J.~Berkowitz,
\newblock {\em Perturbation theory of eigenvalue problems} (CRC Press, 1969).

\bibitem{kato2013perturbation}

\newblock T.~Kato{\em Perturbation theory for linear operators} Vol. 132
  (Springer Science \& Business Media, 2013).

\bibitem{Lizzi:2003ru}
F.~Lizzi, P.~Vitale, and A.~Zampini,
\newblock {\em {The Fuzzy disc}},
\newblock JHEP {\bf 08}, 057 (2003), hep-th/0306247.

\bibitem{lee2013smooth}
J.~M. Lee,
\newblock Smooth manifolds,
\newblock in {\em Introduction to Smooth Manifolds}, pp. 1--31, Springer, 2013.

\bibitem{Seiberg:1999vs}
N.~Seiberg and E.~Witten,
\newblock {\em {String theory and noncommutative geometry}},
\newblock JHEP {\bf 09}, 032 (1999), hep-th/9908142.

\bibitem{Aoki:1999vr}
H.~Aoki {\em et~al.},
\newblock {\em {Noncommutative Yang-Mills in IIB matrix model}},
\newblock Nucl. Phys. {\bf B565}, 176 (2000), hep-th/9908141.

\bibitem{Steinacker:2008ri}
H.~Steinacker,
\newblock {\em {Emergent Gravity and Noncommutative Branes from Yang-Mills
  Matrix Models}},
\newblock Nucl. Phys. B {\bf 810}, 1 (2009), 0806.2032.

\bibitem{Steinacker:2010rh}
H.~Steinacker,
\newblock {\em {Emergent Geometry and Gravity from Matrix Models: an
  Introduction}},
\newblock Class. Quant. Grav. {\bf 27}, 133001 (2010), 1003.4134.

\bibitem{Steinacker:2019fcb}
H.~C. Steinacker,
\newblock {\em {On the quantum structure of space-time, gravity, and higher
  spin}},
\newblock (2019), 1911.03162.

\bibitem{Steinacker:2016nsc}
H.~C. Steinacker,
\newblock {\em {String states, loops and effective actions in noncommutative
  field theory and matrix models}},
\newblock Nucl. Phys. {\bf B910}, 346 (2016), 1606.00646.

\bibitem{Iso:2000ew}
S.~Iso, H.~Kawai, and Y.~Kitazawa,
\newblock {\em {Bilocal fields in noncommutative field theory}},
\newblock Nucl. Phys. {\bf B576}, 375 (2000), hep-th/0001027.

\bibitem{Minwalla:1999px}
S.~Minwalla, M.~Van~Raamsdonk, and N.~Seiberg,
\newblock {\em Noncommutative perturbative dynamics},
\newblock JHEP {\bf 02}, 020 (2000), hep-th/9912072.

\bibitem{voisin2003hodge}

\newblock C.~Voisin{\em Hodge theory and complex algebraic geometry II} Vol.~2
  (Cambridge University Press, 2003).

\bibitem{Baouendi:1999uya}
M.~S. Baouendi, P.~Ebenfelt, and L.~Rothschild,
\newblock {\em {Real submanifolds in complex space and their mappings}},
  Princeton Mathematical Series Vol.~47 (Princeton University Press, Princeton,
  NJ, 1999).

\bibitem{Perelomov:1986tf}
A.~M. Perelomov,
\newblock {\em {Generalized coherent states and their applications}} (Springer,
  Berlin Heidelberg, 1986).

\bibitem{Aschieri:2006uw}
P.~Aschieri, T.~Grammatikopoulos, H.~Steinacker, and G.~Zoupanos,
\newblock {\em {Dynamical generation of fuzzy extra dimensions, dimensional
  reduction and symmetry breaking}},
\newblock JHEP {\bf 09}, 026 (2006), hep-th/0606021.

\bibitem{Aoki:2014cya}
H.~Aoki, J.~Nishimura, and A.~Tsuchiya,
\newblock {\em {Realizing three generations of the Standard Model fermions in
  the type IIB matrix model}},
\newblock JHEP {\bf 05}, 131 (2014), 1401.7848.

\bibitem{Yang:2006dk}
H.~S. Yang,
\newblock {\em {Emergent Gravity from Noncommutative Spacetime}},
\newblock Int. J. Mod. Phys. A {\bf 24}, 4473 (2009), hep-th/0611174.

\bibitem{Aoki:1998bq}
H.~Aoki {\em et~al.},
\newblock {\em {IIB matrix model}},
\newblock Prog. Theor. Phys. Suppl. {\bf 134}, 47 (1999), hep-th/9908038.

\bibitem{Chatzistavrakidis:2011gs}
A.~Chatzistavrakidis, H.~Steinacker, and G.~Zoupanos,
\newblock {\em {Intersecting branes and a standard model realization in matrix
  models}},
\newblock JHEP {\bf 09}, 115 (2011), 1107.0265.

\bibitem{Arnlind:2012cx}
J.~Arnlind and J.~Hoppe,
\newblock {\em {The world as quantized minimal surfaces}},
\newblock Phys. Lett. B {\bf 723}, 397 (2013), 1211.1202.

\bibitem{Lukierski:1993wx}
J.~Lukierski, H.~Ruegg, and W.~J. Zakrzewski,
\newblock {\em {Classical quantum mechanics of free kappa relativistic
  systems}},
\newblock Annals Phys. {\bf 243}, 90 (1995), hep-th/9312153.

\bibitem{Grosse:1993uq}
H.~Grosse and P.~Presnajder,
\newblock {\em {The Construction on noncommutative manifolds using coherent
  states}},
\newblock Lett. Math. Phys. {\bf 28}, 239 (1993).

\bibitem{Sperling:2018hys}
M.~Sperling and H.~C. Steinacker,
\newblock {\em {Intersecting branes, Higgs sector, and chirality from $
  \mathcal{N} $ = 4 SYM with soft SUSY breaking}},
\newblock JHEP {\bf 04}, 116 (2018), 1803.07323.

\bibitem{Sperling:2019xar}
M.~Sperling and H.~C. Steinacker,
\newblock {\em {Covariant cosmological quantum space-time, higher-spin and
  gravity in the IKKT matrix model}},
\newblock JHEP {\bf 07}, 010 (2019), 1901.03522.

\bibitem{Grosse:1996mz}
H.~Grosse, C.~Klimcik, and P.~Presnajder,
\newblock {\em {On finite 4-D quantum field theory in noncommutative
  geometry}},
\newblock Commun. Math. Phys. {\bf 180}, 429 (1996), hep-th/9602115.

\bibitem{Medina:2012cs}
J.~Medina, I.~Huet, D.~O'Connor, and B.~P. Dolan,
\newblock {\em {Scalar and Spinor Field Actions on Fuzzy $S^4$: fuzzy $CP^3$ as
  a $S^2_F$ bundle over $S^4_F$}},
\newblock JHEP {\bf 08}, 070 (2012), 1208.0348.

\bibitem{Sperling:2018xrm}
M.~Sperling and H.~C. Steinacker,
\newblock {\em {The fuzzy 4-hyperboloid $H^4_n$ and higher-spin in Yang--Mills
  matrix models}},
\newblock Nucl. Phys. B {\bf 941}, 680 (2019), 1806.05907.

\bibitem{Hasebe:2012mz}
K.~Hasebe,
\newblock {\em {Non-Compact Hopf Maps and Fuzzy Ultra-Hyperboloids}},
\newblock Nucl. Phys. {\bf B865}, 148 (2012), 1207.1968.

\bibitem{Steinacker:2019awe}
H.~C. Steinacker,
\newblock {\em {Higher-spin kinematics \& no ghosts on quantum space-time in
  Yang-Mills matrix models}},
\newblock (2019), 1910.00839.

\bibitem{Fiore:2019rgy}
G.~Fiore and F.~Pisacane,
\newblock {\em {On localized and coherent states on some new fuzzy spheres}},
\newblock (2019), 1906.01881.

\end{thebibliography}
\bibliographystyle{diss}

\end{document}